\documentclass{article}

\usepackage{times}
\usepackage{hyperref}
\usepackage{url}
\usepackage{amssymb}
\usepackage{amsmath}
\usepackage{amsthm}
\usepackage{color}
\usepackage{xcolor}
\usepackage{fullpage}
\usepackage{xspace}
\usepackage[linesnumbered,ruled,vlined]{algorithm2e}
\usepackage{color}
\usepackage{multirow}
\usepackage{nccmath}
\usepackage[normalem]{ulem}
\usepackage{cite}

\usepackage{xcolor}

\newcommand{\dk}[1]{\textcolor{red}{#1}}

\renewcommand{\dk}[1]{#1}
\newcommand{\parhead}[1]{{\textbf{#1.}\xspace}}

\newcommand{\DN}{{Dynamic Network}\xspace}
\newcommand{\DNs}{{Dynamic Networks}\xspace}

\newcommand{\ADNs}{{Anonymous Dynamic Networks}\xspace}

\newcommand{\MC}{\textsc{Methodical Counting}\xspace}
\newcommand{\MMC}{\textsc{Methodical Multi-Counting}\xspace}

\newcommand{\nameD}{MMC\xspace}

\newcommand{\ldr}{{supervisor}\xspace}
\newcommand{\ldrs}{{{\ldr}s}\xspace}

\newcommand{\notldr}{{supervised}\xspace}
\newcommand{\notldrs}{{\notldr nodes}\xspace}

\newcommand{\cT}{\mathcal{T}}

\newcommand{\probname}{\textsc{All-to-all Communication}\xspace}
\newcommand{\algname}{\textsc{RMC}\xspace}
\newcommand{\algfullname}{\textsc{Restricted Methodical Counting}\xspace}

\newcommand{\allallfull}{\textsc{All-to-All Communication}\xspace}
\newcommand{\mult}{\textsc{Multiplicity}\xspace}

\newcommand{\cc}{Congested Clique\xspace}

\newcommand{\adcs}{Anonymous Dynamic Congested Systems\xspace}
\newcommand{\ADCS}{ADCS\xspace}
\newcommand{\WSM}{Weak Sensor Model\xspace}
\newcommand{\iso}{{isoperimetric~number}\xspace}

\newcommand{\rumor}{{input\ message}\xspace}
\newcommand{\rumors}{{input\ messages}\xspace}

\renewcommand{\vec}[1]{\mathbf{#1}}
\newtheorem{theorem}{Theorem}
\newtheorem{lemma}{Lemma}
\newtheorem{corollary}{Corollary}

\newtheorem{observation}{Observation}

\newtheorem{claim}{Claim}
\newcommand{\remove}[1]{}

\begin{document}

\title{Efficient Distributed Computations in \adcs with Opportunistic Connectivity~\thanks{This work was partially supported by Pace University SRC grant and Kenan fund.}}

\author{
Dariusz R. Kowalski\\
Augusta University,\\ Augusta, GA, USA,\\ 
\href{mailto:dkowalski@augusta.edu}{dkowalski@augusta.edu}
\and
Miguel A. Mosteiro\\
Pace University,\\ New York, NY, USA.\\
\href{mailto:mmosteiro@pace.edu}{mmosteiro@pace.edu}
}

\date{}

\maketitle

\begin{abstract}
In this work we address the question of efficiency of distributed computing in anonymous, congested and highly dynamic and not-always-connected networks/systems. 
More precisely, the system consists of an unknown number of anonymous nodes with congestion on links and local computation. Links can change arbitrarily from round to round, with only limitation that the union of 
any $\cT$ consecutive networks must form a temporarily
connected (multi-)graph on all nodes (knowledge of $\cT$ is the only information the nodes require, otherwise the communication would not be feasible).
Nodes do not have any IDs, only some number $\ell$ of them 
have a bit distinguishing them from nodes without such a bit. In each round a node can send and receive messages from its current neighbors. 
Links and nodes are congested, in the sense that the length of messages and local cache memory for local computation is (asymptotically) logarithmic.

All-to-all communication is a fundamental principle in distributed computing -- it assumes that each node has an \rumor to be delivered to all other nodes. Without loss of generality, the size of each \rumor is logarithmic to fit in the link and node congestion assumption; otherwise, they could be split in logarithmic batches and considered one-by-one. Because of anonymity, each node needs to receive only a set of all \rumors, each accompanied by a number of initiating nodes (message multiplicity). We prove that this task can be done in time polynomial in the (initally unknown) number of nodes $n$ and in the lower bound on the isoperimetric numbers of dynamically evolving graphs.
This allows to efficiently emulate a popular Congested Clique model on top of \adcs (\ADCS) with Opportunistic Connectivity, even if the number of nodes may arbitrarily change in the beginning of emulation.

\end{abstract}


\section{Introduction}
\label{sec:intro}

The \cc~\cite{lotker2005minimum} is a standard synchronous message passing model of distributed computation. In such model there are $n$ labeled nodes that synchronously, in \emph{rounds}, communicate among them. In each round, each node may send a message of $O(\log n)$ bits to each of the other nodes and perform some local computations. Performance is measured in rounds of communication. 
But can we efficiently run \cc algorithms
if the nodes are anonymous (e.g., for the sake of privacy), of restricted capacity (e.g., logarithmic local memory) and their number is unknown?
Even more, what if links may be available only occasionally, depending on an adversarial schedule, and the only necessary 
requirement
is that 
the union of
any $\cT$ consecutive
networks,
for some known parameter $\cT$, forms a temporally connected multi-graph? Could deterministic distributed computing on such Congested Clique be made efficient? In this work we address and answer this question in the affirmative by presenting an \probname algorithm to emulate one round of \cc under a harsh \DN model that includes all those restrictions.

Distributed computations in the \cc model have attracted a lot of attention recently~\cite{becker2015brief,censor2019algebraic,censor2020derandomizing,dolev2012tri,drucker2014power,ghaffari2016improved,ghaffari2016mst,hegeman2015toward,hegeman2015lessons,hegeman2014near,henzinger2019deterministic,korhonen2016deterministic,lenzen2013optimal,nanongkai2014distributed,patt2011round,ghaffari2017distributed}, the reason being that all-to-all communication is becoming a frequent feature of modern distributed systems. However, algorithms developed for such environments are not compatible with many systems where not all nodes are directly connected. Thus, in order to apply the wealth of \cc research to those systems, protocols to emulate a single round of all-to-all communication in congested multi-hop networks are needed.

System restrictions may go far beyond multi-hop communication. Indeed, in some environments network connectivity may be highly dynamic due to mobility or unreliability (e.g. \DNs~\cite{kuhn2010distributed,gilbert2010r,kuhn2011dynamic}). Also, node identifiers may not be feasible in massive low-cost platforms, or one may not want to reveal the identifiers due to privacy concerns (e.g. \ADNs~\cite{spirakis,KowalskiMicalp18journal,KowalskiMicalp19journal,conscious,LunaB15}). 
Additionally, in a system where message size is limited (as in the \cc) it is natural to apply the same limitation to memory access, specially when nodes are expected to be low-cost devices (e.g. \WSM~\cite{FCFM:hoptimalityJournal,FMT:detselection}).

Following up on the CONGEST model~\cite{peleg2000distributed}, in the \cc nodes initially know their neighbors. Being a clique implies that they know the total number of nodes. 
However, when implementing the Congested Clique in a multi-hop topology, nodes may not be connected to all other nodes. 
Moreover, in face of dynamicity and anonymity, each node does not even know the number of its neighbors before receiving messages from them. 
Thus, in this work we assume that neither the exact number, nor even an upper bound on the total number of nodes, is initially known. 
Even more, the network does not need to be connected at all, as long as a union of $\cT$ consecutive networks is temporally connected, for some given parameter $\cT$.

To the best of our knowledge, our \adcs model described above is the most challenging for distributed computing among the existing models in the \DNs literature.

\parhead{Our Contributions and Approach}
In this work, we present a deterministic protocol to emulate a single round of the \cc on \adcs (\ADCS) with Opportunistic Connectivity. That is, on systems where nodes lack identifiers, message size and memory access are limited to $O(\log n)$ bits, where $n$ is the number of nodes (initially unknown), and the communication network is multi-hop and adversarially dynamic, allowing even disconnection with some limitations. The overhead introduced by our emulator is a polynomial function of $n$ and, if known, a lower bound $i_{min}$ on the isoperimetric numbers of dynamic networks,\footnote{%
Formally, it is a lower bound on an isoperimetric number of the product of $\cT$ consecutive networks, as defined formally in Section~\ref{sec:convtime}.
}
making our bounds tighter for networks with good expansion -- more precisely, the $i_{min}$ in the denominator of the time complexity formula could be as large as $\Theta(n)$ for networks with good expansion. See Table~\ref{tab:results} for details and comparison with most closely related work.
We also prove that the knowledge of the connectivity parameter $\cT$ and the number $\ell$ of distinguished nodes (they have additional distinguishing bit, and are called \emph{\ldrs} throughout the paper) is necessary for \probname 
in \ADCS with Opportunistic Connectivity. Similarly, using $o(\log n)$ local memory bits are not enough for \probname in the model.

\begin{table}[htbp]
    \centering
    \begin{tabular}{c|c|c|c|c|c|c}
        Ref & \begin{tabular}{c}Known\\$n$\end{tabular} & Connected & \begin{tabular}{c}Message\\size\end{tabular} & \begin{tabular}{c}Memory\\access\end{tabular} & 
        IDs
        & Time complexity\\
        \hline
        \begin{tabular}{c}this\\ work\end{tabular} & no & no & $O(\log n)$ & $O(\log n)$ & no & $O\left(\frac{n^{1+2\cT(1+\epsilon)} }{\ell i_{\min}^2} \log n \log\left(\frac{n}{\ell}\right)+\frac{ n'}{ \ln(1+i_{\min})} \log^2 n\right)$ \\
        \hline
        \cite{KuhnLO2010} & no & yes & $O(\log n)$ & no limit & yes & \begin{tabular}{c}$O(n^2)$\\(1-Interval connected)\end{tabular} \\
        \cline{2-7}
        & no & yes & $O(\log n)$ & no limit & yes & \begin{tabular}{c}$\Omega(n\log n)$\\(1-Interval connected,\\centralized algorithm)\end{tabular} \\
        \hline
        \cite{haeupler2011faster} & \begin{tabular}{c}upper\\bound\end{tabular} & yes & $b\geq\log n$ & no limit & yes & \begin{tabular}{c}$O(dn^2/b^2)$\\($b\geq d\geq \log n$)\end{tabular}\\
    \end{tabular}\ 
    \caption{
    Comparison of most relevant deterministic \allallfull results in \DNs. Message size and memory access are in bits, and time complexity is in rounds of communication. $n$ denotes the number of nodes (initially unknown), $n'$ -- the number of different \rumors, $\ell$ -- the number of nodes with distinguishing bit (supervisors), $\cT$ -- the connectivity parameter, $\epsilon>0$ -- any chosen constant, $i_{min}$ -- the lower bound on the isoperimetric numbers of the \ADCS as defined in Section~\ref{sec:convtime} 
    (if unknown, the formula holds after substituting 
    $2/n$ for $i_{\min}$).
    $d$ is the size of the \rumor in~\cite{haeupler2011faster}.
    }
    \label{tab:results}
\end{table}

Our algorithmic approach is not a simple classic gossip-based
technique, which is well-known in Distributed Computing~\cite{KowalskiMicalp18journal, FMT:aggJournal, mihail, KDGgossip, spirakis,giakkoupis2011tight}. 
In such algorithms initially nodes hold some values to be shared with neighboring nodes repeatedly until some stopping rule is met.
In fact, the ADCS model yields such techniques incorrect, as multiplicities of the same \rumor may not be counted correctly due to anonymity and dynamic behavior of the underlying networks. Therefore, we combine a spreading-with-stopping technique with coding messages and distribution of potential,~c.f.,~some~of their applications in less-demanding systems~\cite{KowalskiMicalp18journal, FMT:aggJournal, mihail}.
Combining them in a carefully selected way results in a complex algorithm, which occurs surprisingly efficient in very demanding systems such as ADCS with Opportunistic~Connectivity.

One of the critical components of our \allallfull algorithm is based on distribution of potential. We use it when counting all active nodes or nodes with an identified \rumor. The main challenge is to show that the potential distribution process stabilizes (with negligible deviations at nodes) at a desired value -- unlike a typical mass-distribution process, c.f.,~\cite{SJ:count, michal, mihail,sauerwald}, our process has to accommodate (1) lack of connectivity (only a union of $\cT$ consecutive networks needs to be connected), (2) truncated potentials (we need to truncate the numbers due to congested links/nodes), and (3) lack of any knowledge of $n$ (in the beginning).

\parhead{Roadmap} 
In Section~\ref{sec:relwork} we discuss the related work. Section~\ref{sec:model} presents the model and useful definitions, while Section~\ref{sec:prel} states preliminary results on imposibility of all-to-all communication and properties of dynamic networks.
The main \allallfull algorithm is presented and analyzed in Section~\ref{sec:all-to-all}, while its major components, \mult and \algfullname, are given and analyzed in Sections~\ref{sec:mult-algorithm} and~\ref{sec:counting}, respectively.

\section{Related Work}
\label{sec:relwork}

To the best of our knowledge, there is no previous study of distributed computations under a \DNs model that include all the restrictions of \adcs. That is, worst-case dynamicity with connectivity that is only opportunistic, communication limited to $O(\log n)$ bits per round, local memory access also limited to $O(\log n)$ bits per round, and nodes without ID's.
We overview in this section the \DNs extant work on related models.  

Two closely related works on \allallfull in \DNs are~\cite{KuhnLO2010,haeupler2011faster}. In both, the model includes node ID's, continuous connectivity, and unbounded memory access. Nevertheless, lower bounds under weaker conditions apply to the \ADCS stricter model. For deterministic algorithms, a lower bound of $\Omega(n\log n)$ rounds was proved in~\cite{KuhnLO2010}, even for centralized algorithms.  

The most frequent model of worst-case dynamicity still assumes continuous connectivity. That is, even though the set of links may change arbitrarily from round to round, it is guaranteed that in every round of communication there is a path between every pair of nodes. For instance, in the population protocol model in~\cite{AngluinADFP2006} and the continuous connectivity model in~\cite{OW:diss}. Later on, in~\cite{KuhnLO2010}, Kuhn, Lynch and Oshman parameterized the continuous connectivity assumption in the $T$-interval connectivity model where the topology may change, but changes are restricted to maintain an underlying connected static graph over each sequence of $T$ rounds. Instantiating $T=1$ the model is equivalent to continuous connectivity. 

In none of the above models disconnection is allowed, not even temporarily.  
In a more recent work~\cite{michail2014causality}, Michail, Chatzigiannakis, and Spirakis do consider disconnections, with limitations that yield opportunistic connectivity. Specifically, in at most $k$ rounds every node \emph{influences} at least one other node that has not been influenced yet (meaning for example that the status of a node is learned by some other node in at most $k$ rounds). Since the limitation holds for every node, the model is more restrictive than simply parameterizing the overhead on dissemination of one message to a factor of $k\geq 1$. In fact the authors show that for $k=1$ it is not possible to influence only one node in every round. 
In the same paper, the authors study a second model of opportunistic connectivity where communication between neighboring nodes must be allowed within some time window, but the underlying topology is fixed.
The \probname problem is not studied in that work.

Communication congestion has been thoroughly studied in the \cc model and others~\cite{becker2015brief,censor2019algebraic,censor2020derandomizing,dolev2012tri,drucker2014power,ghaffari2016improved,ghaffari2016mst,hegeman2015toward,hegeman2015lessons,hegeman2014near,henzinger2019deterministic,korhonen2016deterministic,lenzen2013optimal,nanongkai2014distributed,patt2011round,ghaffari2017distributed}, but to the best of our knowledge ours is the first model to combine communication and memory access congestion. 

Finally, with respect to \DNs with anonymous nodes, the \ADNs (ADN) model has attracted also a lot of attention recently~\cite{spirakis,KowalskiMicalp18journal,KowalskiMicalp19journal,conscious,oracle,experimentalConscious,LunaB15,opodisCounting, netysCounting, spaa21consensus}. A comprehensive overview of work related to ADNs can be found in a survey by Casteigts, Flocchini, Quattrociocchi, and Santoro~\cite{arnaudSurvey}. In all these works continuous connectivity is assumed. 

Other studies also dealing with the time complexity of information gathering exist~\cite{
chen2013role,
banerjee2014epidemic,
sanghavi2007gossiping,
boyd2006randomized,
mosk2008fast,
sarma2015distributed}, 
but include in their model additional assumptions, such as the network having the same topology frequently enough or node identifiers.

\section{Problem and Model}
\label{sec:model}

To emulate a single round of communication in the \cc, we study the \probname problem defined as follows. Initially all nodes hold an \rumor, and to solve the problem all nodes must receive the \rumor of all other nodes. 
For each \rumor initially held by  more than one node, nodes must receive the number of copies of that \rumor as well.

We study \probname in \adcs (\ADCS) formed by $n$ processing \emph{nodes}. 
Nodes lack identifiers and the set of communication links among nodes is adversarially dynamic, with some limits on disconnection as specified below. 
When two nodes are able to communicate (that is, they are the endpoints of a communication link), we say that they are \emph{neighbors}. 
We assume that $n$ is initially unknown.  
Moreover, due to dynamicity, each \ADCS node does not know even the number of its neighbors before receiving messages from them. 
That is, a node knows the number of its neighbors in any one round after receiving messages, but it does not who they are. Nonetheless, although we refer to the set of neighbors of a node and we label nodes in the presentation of algorithms and the analysis, we do it only for the sake of clarity -- the nodes have only access to the messages sent by neighbors.
Throughout the paper, references to algorithms lines are given as $\langle algorithm\#\rangle.\langle line\#\rangle$ for succinctness.

Time is discretized in communication \emph{rounds}. In each round of communication, each node may send a \emph{message} to all its neighbors (the same to all\footnote{Dynamicity and anonimity prevent the nodes from sending destination-oriented messages. 
}),
receive messages from all its neighbors, access local memory, and perform local computations. 
We evaluate time complexity in rounds of communication, given that in comparison local memory access and local computations take negligible time. 

Each message is limited to $O(\log n)$ bits, as in the \cc~\cite{lotker2005minimum} and CONGEST~\cite{peleg2000distributed} models. 
Additionally, each memory access is limited to $O(\log n)$ bits as well -- this is the reason we call nodes {\em congested} as well. Specifically, in each side of a communication link there is an input buffer and an output buffer, both limited to $O(\log n)$ bits. 
The additional internal memory for local computation also holds $O(\log n)$ bits.
It implies, in particular, that a local computation algorithm may process a constant number of received messages (in their buffers) at a time, aggregating them somehow in local memory, and continuing with other buffers. (In other words, in a single round the algorithm may check all messages, but cannot upload them from the buffers all at once but instead needs to process them in constant batches). At the end of the round, the algorithm may store some $O(\log n)$-bit information in external memory, e.g., the learned \rumor together with its multiplicity.
We show below that this restriction is tight. That is, if memory access is restricted to $o(\log n)$ bits, some instances of \probname cannot be solved. 

We assume the presence of $\ell$ distinguished nodes, called $\ldrs$, where $0<\ell<n$. That is, the set of nodes is partitioned in two classes: \ldrs and \notldrs. The number of \ldrs $\ell$ is known to all nodes. As we show later, these assumptions are necessary to solve \probname deterministically without knowing $n$.
Within each class, nodes are indistinguishable.

The communication network topology model is a $\cT$-connected time-evolving graph, defined as follows\footnote{There are many formalisms in the literature to specify graphs that change with time. Names include temporal graphs, dynamic graphs, evolving graphs, time-varying graphs, and others. We adopt the notation that provides more clarity to our analysis.}.  

Given a fixed set $V$ of $n$ nodes, let a \emph{time-evolving graph} (or \emph{evolving graph} for short) be an infinite sequence of graphs $\mathcal{G}=\{G^{(t)}\}_{t\in \mathbb{N}}$
 such that $G^{(t)}=\left(V,E^{(t)}\right)$, where each $E^{(t)}$ is the (possibly different) set of links of the graph $G^{(t)}$. We call each $G^{(t)}$ a \emph{constituent graph}.

For any pair of nodes $u,v \in V$ in an evolving graph $\mathcal{G}$, let an \emph{opportunistic path} of length $k>0$ from $u$ to $v$ be a sequence of links $(x_1,x_2),(x_2,x_3),\dots,(x_{k-1},x_k)$ where $x_i\in V$ for every $i\in[k]$, $x_1=u$, $x_k=v$, and for every consecutive pair of links $(x_{i-1},x_{i}),(x_{i},x_{i+1})$ in the sequence, such that $1 < i < k$, if $(x_{i-1},x_{i})\in E^{(t)}$ and $(x_{i},x_{i+1})\in E^{(s)}$, then it is $t\leq s$. In words, an opportunistic path from $u$ to $v$ is a path that may not exist in its entirety at any round, but it can be traversed from $u$ to $v$ throughout time (sometimes called a \emph{journey} in the literature). 

We say that an evolving graph $\mathcal{G}$ is \emph{$\cT$-connected} if, for each $t\geq 0$, for each pair of nodes $u,v \in V$, and for each consecutive sequence of $\cT$ constituent graphs $G^{(t+1)},\dots,G^{(t+\cT)}$, there is an opportunistic path from $u$ to $v$.  
Note that each $G^{(t)}$ may be disconnected, and that we do not restrict the length of the opportunistic path.
As we show later, \probname is not possible in \ADCS without knowledge of $\cT$. Thus, we assume that nodes know $\cT$.

With respect to the expansion properties of the \ADCS  network, we consider three scenarios: no knowledge, knowledge of minimum 
conductance, and knowledge of minimum \iso. The definitions of these expansion characteristics are given in Section~\ref{sec:convtime}.

\section{Preliminaries}
\label{sec:prel}

\subsection{Impossibility Results}

We start this section establishing some impossibility facts that validate the assumptions of our model. 

\begin{observation}
For each \probname deterministic algorithm $\mathcal{A}$, there exists an \ADCS with 
$\cT$-connected evolving graph topology such that, if $\mathcal{A}$ does not use $\cT$, the problem cannot be solved. \end{observation}
\begin{proof}
For the sake of contradiction, assume there exists such algorithm $\mathcal{A}$.
Let $T$ be the worst-case running time of $\mathcal{A}$ on a $1$-connected \ADCS of $n$ nodes. 
Consider an \ADCS with $n+1$ nodes that during the first $T$ rounds of execution of $\mathcal{A}$ is formed by a clique of $n$ nodes and an isolated node $v$, which connects to the clique for round $T+1$.
The nodes in the clique are not able to communicate with $v$ during the $T$ execution steps of $\mathcal{A}$. Hence they stop, but then the \probname problem was not solved  (as nodes in the clique do not know the message of the other node), and the described execution is feasible for $\cT$-connected \ADCS, where $\cT=T+1$ is unknown to the algorithm. 
\end{proof}

\begin{observation}
There is no deterministic algorithm to solve \probname in an \ADCS without at least one distinguished node, and without knowledge of the total number of nodes $n$ and the number of distinguished nodes $\ell$.
\end{observation}
\begin{proof}
For the sake of contradiction, assume there exists such algorithm, call it $\mathcal{A}$. Then, we can use $\mathcal{A}$ to compute $n$ by simply assigning a message $1$ to each node and counting how many $1$'s are received. However, this is a contradiction because, even
for \ADNs without node and edge congestion
it has been shown in~\cite{spirakis} that counting the number of nodes deterministically is not possible without some distinguished node. It was also shown in~\cite{KowalskiMicalp19journal} that the number of distinguished nodes needs to be known. 
\end{proof}

\begin{observation}
There exist applications of \probname where, if memory access is limited to $o(\log n)$ bits, the problem cannot be solved.
\end{observation}

\begin{proof}
Consider an \probname algorithm used to compute the number of nodes $n$ in an \ADCS, where $n$ is unknown 
(i.e., each node has initially the same message $1$). For each $n$ there should be at least one final state of a node in which the node stops and outputs $n$. On the other hand, nodes with internal memory $\mu=o(\log n)$ may result in only $2^\mu = o(n)$ states; hence, for some $n$ there will be no terminating state outputting $n$ and the \probname fails.
\end{proof}

\subsection{Expansion in Evolving Graphs with Opportunistic Connectivity}
\label{sec:convtime}

The relevant property that $\cT$-connectivity provides to our analysis is the following. 
\begin{observation}
\label{obs:sc}
Given a $\cT$-connected evolving graph $\mathcal{G}$, for any sequence of constituent graphs of $G^{(t+1)},G^{(t+2)},\dots,G^{(t+\cT)}$, $t\geq0$, the union graph $G_{\cup_t}=\left(V,\cup_{i=1}^{\cT} E^{(t\cT+i)}\right)$ is connected.  \footnote{Notice that the one or more paths between each pair of nodes $u,v\in V$ in $G_{\cup_t}$ are derived from the opportunistic path from $u$ to $v$ and the opportunistic path from $v$ to $u$, which exist due to $\cT$-connectivity.} 
\end{observation} 

Consider a $\cT$-connected evolving graph $\mathcal{G}$ conceptually divided in subsequences of $\cT$ consecutive constituent graphs. That is, 
$\mathcal{G} = \{\mathcal{G}_i\}_{i=0}^{\infty}$,  where 
$\mathcal{G}_i = \{G_{i\cT+t}\}_{t=1}^{\cT}$.
By Observation~\ref{obs:sc}, for each $i\geq 0$, the union graph $G_{\cup_i}$ defined on $\mathcal{G}_i$ is connected, whereas each constituent graph of $\mathcal{G}$ may be not connected.
Thus, rather than analyzing the potential distribution process as a Markov chain on the evolving graph $\mathcal{G}$, we study the process on the evolving graph $\mathcal{G^*}=\{G_{\cup_i}\}_{i=0}^{\infty}$, where each constituent union graph $G_{\cup_i}$  is connected. We do so taking into account that, in fact, each union graph corresponds to a sequence of graphs from $\mathcal{G}$.

The potential distribution process can be seen as a multiplication of a vector of values, one component for each node, by a matrix of shares, where component $(u,v)$ corresponds to the fraction of potential shared by node $u$ with neighboring node $v$. 
Let $P^{(t)}$ be the matrix of shares used by a gossip-based algorithm in round $t$ (corresponding to a constituent graph $G^{(t)}$). Then, for each $i\geq 0$, the matrix of shares corresponding to the potential distribution on the evolving graph $\mathcal{G}_i$ is $P_i = \prod_{t=1}^{\cT} P^{(i\cT+t)}$.
Each $P^{(t)}$ is doubly stochastic, hence, $P_i$ is also doubly-stochastic. 

For each $i\geq 0$, consider a time-homogeneous Markov chain $\mathbf{X}_i$ with state space $V$ and transition matrix $P_i$.
$\mathbf{X}_i$ is finite, irreducible and aperiodic, and given that each $G_{\cup_i}$  is connected, it is ergodic.
Thus, the stationary distribution of $\mathbf{X}_i$ is unique~\cite{book:motwani}. The uniform distribution, that is, $\pi_v=1/n$ for all $v\in V$, is a solution of $\vec\pi P_i  = \vec\pi$ because $P_i$ is doubly stochastic. Thus, the stationary distribution of $\mathbf{X_i}$ is uniform. 

Classic bounds on mixing time of Markov chains require the transition matrix to be symmetric~\cite{SJ:count}, but the $P_i$'s above may not be. Thus, we apply instead a bound by Mihail~\cite{mihail} that is applicable to arbitrary irreducible Markov chains, as long as they are strongly aperiodic (as in our case where nodes keep at least half of the potential), and it is a function of transition matrix conductance. 

The \emph{conductance} of a transition matrix $P=\big(p_{uv}\big)$ of a Markov chain over state space $V$ with stationary distribution $\vec{\pi}$ is defined as follows: 
\begin{align*}
\phi(P) 
=
\min_{S\subset V: \sum_{v\in S}\vec{\pi}_v\leq 1/2} \phi_{P}(S)
\ , \ \ \ \ \ \ \ \ where \ \ \ \ \ \ \ \ \ 
\phi_{P}(S) 
=
\frac{\sum_{u_j\in S}\sum_{u_h\in V\setminus S}w_{u_j u_h}}{\sum_{v\in S}\vec{\pi}_v}
\ \ \ \ \ \ \ and \ \ \  \ \ \ \ 
w_{u_ju_h} 
=
\pi_{u_j} p_{u_ju_h} 
\ .
\end{align*}
We instantiate this definition on $P_i=\big(p_{uv}\big)$ and $\vec{\pi}=\frac{\vec{I}}{n}$ as follows:
\begin{align*}
\phi(P_i) 
=
\min_{S\subset V: |S|\leq n/2} \phi_{P_i}(S)
\ , \ \ \ \ \ \ \ \ \ \ \ where \ \ \ \ \ \ \ \ \ \ \ \ \ 
\phi_{P_i}(S) 
=
\frac{1}{|S|} \sum_{u_j\in S}\sum_{u_h\in V\setminus S} p_{u_ju_h}
\ .
\end{align*}

\vspace*{-3ex}
\begin{theorem}
\label{thm:mihail}
(derived from Theorem 3.1 in~\cite{mihail}.)
For an irreducible strongly aperiodic Markov chain $\mathbf{X}_i$ with state space~$V$, such that $|V|=n$, with a unique uniform stationary distribution and transition matrix $P_i$, it is
\begin{align*}
\left|\left|\vec{\Pi}_{t+1} - \frac{\vec{I}}{n}\right|\right|_2^2 &\leq \left(1-\phi(P_i)^2\right) \left|\left|\vec{\Pi}_t - \frac{\vec{I}}{n}\right|\right|_2^2 \ .
\end{align*}
Where $\vec{\Pi}_t$ is the distribution after $t\geq0$ steps of $\mathbf{X}_i$.
\end{theorem}

Notice that the above analysis applies to each $i\geq 0$. That is, each $P_i$ may be different for each $\mathcal{G}_i$, but all of the $\mathbf{X_i}$ converge to a uniform stationary distribution and the bounds in Theorem~\ref{thm:mihail} apply to each $\mathbf{X_i}$. Thus, by application of these bounds to each sequence of $\cT$ rounds, the convergence time for the evolving graph $\mathcal{G}$ can be obtained.

We define the minimum conductance $\phi_{\min}$ corresponding to the transition matrices $P_i = \prod_{t=1}^{\cT} P^{(i\cT+t)}$ used by our algorithms on the evolving graph topology $\mathcal{G}$ as follows: 
\[
\phi_{\min} = \min_{i=0,1,2,\dots} \phi(P_i) 
\ .
\]
Conductance is a useful expansion characteristic to provide tighter time bounds, but being a function of the probabilities of transition it is specific for each algorithm. 
To obtain bounds that depend on network characteristics only, we will use the \iso.\footnote{The \iso of a graph (a.k.a. graph Cheeger constant) is the discrete analogue of the Cheeger isoperimetric constant,
c.f.,~\cite{buser1984bipartition}.}
The \iso of a static graph $G=(V,E)$ is defined as follows:
\begin{align*}
i(G) = \min_{X:|X|\leq|V|/2} \frac{|\partial X|}{|X|}
\ ,
\end{align*}
where $\partial X$ denotes the set of links of $G$ that have one end in $X$ and the other end in $V\setminus X$.
We apply this definition to each constituent union graph $G_{\cup_j}$ of the evolving graph
$\mathcal{G^*}=\{G_{\cup_j}\}_{j=0}^{\infty}$, and obtain:
\begin{align*}
i(G_{\cup_j}) = \min_{X:|X|\leq|V|/2} \frac{|\partial X|}{|X|} \text{\ \  for each } j=0,1,\dots 
\ \ \ \ \ \ \ \ \ \ \  \ \text{ and then } \ \ \ \ \ \ \ \ \ \ \ \ \  
i_{\min} = \min_{j=0,1,2,\dots} i(G_{\cup_j})
\ .
\end{align*}

The following relation between conductance and \iso will be used. 
For each $P^{(t)}$, let each non-zero entry be at least $1/d$ for some $d>0$. Given that $P_j = \prod_{t=1}^{\cT} P^{(j\cT+t)}$, we  have that each non-zero entry of $P_j$ is at least $1/d^{\cT}$. 
Thus, for each $j=0,1,\dots$ 
we have
\begin{align}
\phi(P_j) 
&\geq \min_{X\subset V: |X|\leq |V|/2}\frac{1}{|X|} \sum_{u_i\in X}\sum_{u_j\in V\setminus X}\frac{1}{d^{\cT}}
= \frac{1}{d^{\cT}} \min_{X\subset V: |X|\leq |V|/2}\frac{|\partial X|}{|X|}
= \frac{i(G_{\cup_j})}{d^{\cT}}
\ . \label{condvsisop}
\end{align}

\section{All-to-all Communication}
\label{sec:all-to-all}

Assume each node $v\in V$ has an \rumor to send to every other node consisting of at most $\ell_v\le \lceil \log n \rceil$ bits
(recall that $n$ is unknown to nodes, but this is only an upper bound on the \rumor length). This length could be scaled by any constant $c>0$ w.l.o.g.
The set of neighbors $N$ and nodes are labeled only for presentation, but nodes do not have and do not use such knowledge; 
instead, they can only send and receive messages to/from them, see the pseudo-codes.

Our \allallfull Algorithm~\ref{allallAlg} starts with counting the number of participants $n$ using the \algfullname algorithm (refer to Section~\ref{sec:counting}), and then proceeds in subsequent epochs. Each epoch is dedicated to finding a new \rumor
(see the Discovery Part) and counting how many copies of this \rumor are at nodes (see the Processing Part).
The Discovery Part proceeds in phases. Each phase parameterized by $index$ corresponds to the discovery of the $index$ bit of some \rumor, whose first $index-1$ bits  have been discovered in preceding phases (but 
of an \rumor not been discovered and counted in previous epochs). 
This discovery is done by spreading and updating variables $match1$ at each node, 
and if no node is discovered, then spreading and updating variables $match0$. 
The phases' parameter $index$ is numbered from $1$ to $\lceil\log n\rceil$ (the latter being an upper bound on the \rumor's length).
In each phase, a sufficiently long broadcast is run to update information on whether the $index$ bit being matched is 1 or 0. 
The length of the broadcast $r' = \lceil\cT\ln n / \ln(1+i_{\min})\rceil$ 
is a bound on temporal diameter of the evolving graph $\mathcal{G^*}$ based on 
the initially computed
$n$, the connectivity parameter $\cT$, and, 
if known, a bound $i_{\min}$ on the \iso, which otherwise can be lower bounded by a function of $n$  
(refer to the first line in Algorithm~\ref{allallAlg} and Section~\ref{sec:convtime}).
Nodes whose \rumors have been discovered in previous epochs, or whose \rumor does not match the bits already discovered in the current epoch (which is encoded by setting the variable $match$ to $false$) only forward the received information in the main algorithm and the procedure \mult (see Algorithm~\ref{multAlg} in Section~\ref{sec:mult-algorithm}) run in the Processing Part. 
All nodes whose \rumor has been discovered (and so, also delivered to all nodes) in the Discovery Part of the current epoch (which is encoded by setting the variable $delivered$ to $true$) are counted using the procedure \mult by all nodes. The output -- the discovered \rumor and its multiplicity -- is stored in external storage at the end of the epoch. The algorithm finishes when no new \rumor is discovered.

\begin{algorithm}[htbp]
\caption{\allallfull algorithm for each node. $\rumor$ is the \rumor initially held by this node. $N$ is the set of neighbors of this node in the current round -- the node does not know them, but can send and receive short messages to/from them.
$r' = \lceil\cT\ln n / \ln(1+i_{\min})\rceil$ as in Lemma~\ref{lemma:broadcast}, where
$i_{\min}$ is the minimum \iso of the evolving graph topology as defined in Section~\ref{sec:convtime}. 
}
\label{allallAlg}
\DontPrintSemicolon
	$n \gets $ \algfullname() \tcp*{Counting n}
    $delivered \gets false$ \tcp*{\dk{Own \rumor not delivered yet}}
	\While(\tcp*[f]{\dk{Iterating Epochs -- as long as there are undelivered \rumors}}){true \label{while}}
	{
        $match \gets \lnot delivered$ \tcp*{\dk{$match$ will indicate if discovered bits match own \rumor}}
        \tcp{Discovery \dk{Part}}
        $new\ \rumor\gets $ empty string of bits\;
        \For(\tcp*[f]{\dk{Iterating Phases -- from most to least significant bit}}){$index \gets 1$ to $\lceil\log n\rceil$\label{outerloop}}{
            $bit \gets $ extract bit $index$ from $\rumor$\;
            $match0 \gets match \land (bit = 0)$
             \tcp*{Indicates matching so far, with 0 in \dk{$index$ bit}}
            $match1 \gets match \land (bit = 1)$
             \tcp*{Indicates matching so far, with 1 in \dk{$index$ bit}}
            \For(\tcp*[f]{\dk{Broadcasting $match1$}}){$round \gets 1$ to $r'$}{
                Send $\langle match1\rangle$ 
				and Receive $\langle match1_v\rangle , \forall v\in N$ \label{send1} \tcp*{To/from neighbors in $N$}
				$match1 \gets \bigvee_{v\in N} match1_v \lor match1$
				\tcp*{\dk{Incorporating neighbors' knowledge about $match1$}} 
            }
			\eIf(\tcp*[f]{\dk{Discovered $index$ bit 1}}){$match1 = true$\label{match1}}{
			    append $'1'$ to $new\ \rumor$\;
			    \lIf{$bit=0$}{$match \gets false$}
			}{
                \For(\tcp*[f]{\dk{Broadcasting $match0$}}){$round \gets 1$ to $r'$}{
                    Send $\langle match0\rangle$ 
		    		and Receive $\langle match0_v\rangle , \forall v\in N$ \label{send0} \tcp*{To/from neighbors in $N$}
				    $match0 \gets \bigvee_{v\in N} match0_v \lor match0$
				    \tcp*{\dk{Incorporating~neighbors'~knowledge~on~$match0$}}
		    	}
			    \If(\tcp*[f]{\dk{Discovered $index$ bit 0}}){$match0 = true$\label{match0}}{
			        append $'0'$ to $new\ \rumor$\;
			        \lIf{$bit=1$}{$match \gets false$}
			    }
			    \lElse(\tcp*[f]{No more matches, all \rumors delivered}){
			        {\bf return}
			    }
			}
        }
        \tcp{Processing \dk{Part}}
        \lIf(\tcp*[f]{\dk{Own \rumor was delivered in this~epoch}}){$match=true$}{$delivered \gets true$} 
		$count\gets $ \mult($new\ \rumor$)\;
		{\bf output} $\langle new\ \rumor, count\rangle$ \tcp*{Move \dk{newly discovered info} to external storage}
	}
\end{algorithm}

\subsection{Analysis of \allallfull}

In this section we present the main theorem of this work, showing the correctness and running time of our \allallfull algorithm.
For the purpose of the analysis, we conceptually divide time into blocks of $\cT$ rounds.
The analysis of \algfullname(\ ) and \mult(\ ) is included in the sections that follow. 
The following lemma will be used. 

\begin{lemma}
\label{lemma:broadcast}
Consider an \ADCS 
with a $\cT$-connected evolving graph topology with $n$ nodes, each holding a true/false value. 
Then, if all nodes broadcast true values (initially held or received from others) for
at least $r' = \cT\ln n / \ln(1+i_{\min})$ rounds,
all nodes know whether there was initially some true value in the system or not.
\end{lemma}

\begin{proof}
Let $\{S_i,V\setminus S_i\}$ be a partition of the set of nodes at the beginning of some block $i$ 
with transition matrix $P_i = \prod_{t=1}^{\cT} P^{(i\cT+t)}$. 
By definition of \iso (refer to Section~\ref{sec:convtime}) we know that
$i_{\min} \leq |\partial S_i|/|S_i|$.
That is, the number of links crossing the partition $\{S_i,V\setminus S_i\}$ is larger than $|S_i| i_{\min}$.
Hence, adding a 1-hop neighbourhood to any initial subset of nodes $|S_i|$, the new subset at the beginning of block $i+1$ is such that
$|S_{i+1}| \geq (1+i_{\min}) |S_i|$.

Consider any node that did not receive true yet, call it $x$, and let $S_i=\{x\}$. 
The question of how many blocks are needed for $x$ to receive a true (if there is any) is equivalent to ask what is the minimum $t$ such that the $t$-hop neighbourhood of $x$ (in the evolving graph) includes all nodes.
That is, we want to find what is the minimum~$r'$
such~that 
$(1+i_{\min})^{r'} \geq n$. 
Manipulating the latter equation, and taking into account that each block has $\cT$ rounds, the claimed 
$r' = \cT\ln n / \ln(1+i_{\min})$ follows.
\end{proof}


\begin{theorem}
\label{thm:a2a}
Consider an \ADCS 
with a $\cT$-connected evolving graph topology, $\cT\in O(1)$,
with minimum \iso $i_{\min}$,
formed by $\ell\geq 1$ \ldr  nodes and $n-\ell$ \notldr nodes, $n\geq 2$,
each holding an \rumor, 
running the \allallfull algorithm. 
Then, 
the \probname problem is solved in 
$$
O\left( 
\frac{n^{1+2\cT(1+\epsilon)} }
{\ell i_{\min}^2} \log n \log \left(\frac{n}{\ell}\right)
+
\frac{ n'}{ \ln(1+i_{\min})} \log^2 n
\right)\textrm{  rounds,}
$$
where $n'\leq n$ is the number of different \rumors.
\end{theorem}

\begin{proof}
The correctness of \algfullname(\ ) and \mult(\ ) is proved in Theorems~\ref{thm:many} and~\ref{thm:mult}, respectively.
Thus, to complete the proof of correctness, it is enough to prove that all \rumors are discovered, that every discovered \rumor is an \rumor held by some node, and that the algorithm runs under the restrictions of the \ADCS model. Refer to Algorithm~\ref{allallAlg}.

Assume first for the sake of contradiction that there is some node $v$ with \rumor $m$ that is not discovered. 
Each bit of each newly discovered \rumor $m'$ is matched to the \rumors of all other nodes holding $m'$ in Lines~\ref{allallAlg}.\ref{match1} (for a 1 bit) or~\ref{allallAlg}.\ref{match0} (for a 0 bit).
Thus, $m$ will be discovered as long as the number of rounds $r'$ is enough to broadcast the match status of $v$ to all other nodes. 
Lemma~\ref{lemma:broadcast} proves that $r'$ is large enough for $v$ to disseminate its match. Therefore $m$ is discovered, which is a contradiction. 

Assume now, again for the sake of contradiction, that there is some \rumor $m$ that is discovered but it is not an \rumor of any node in the system. 
However, an \rumor is discovered when 
there is a $match$
with the \rumor of some node (or nodes). Therefore, if $m$ is not the \rumor of any node it cannot be discovered.

\algfullname(\ ) and \mult(\ ) are proved to run under the restrictions of the \ADCS model in Theorems~\ref{thm:many} and~\ref{thm:mult}, respectively, and the rest of the messages and calculations in the rest of the algorithm are on Boolean variables. Thus, \algfullname does not violate the $O(\log n)$ bits limit on message size and local computation
All the analysis of the various parts of the algorithm apply to $\cT$-connected evolving graph topologies. Therefore, the proof of correctness is complete.

Regarding the running time complexity, recall that performance is measured in rounds of communication. Theorems~\ref{thm:many} and~\ref{thm:mult} prove the time complexity of \algfullname(\ ) and \mult(\ ) respectively. The running time of the remaining parts can be obtained by simple inspection. Namely, communication is carried out in Lines~\ref{allallAlg}.\ref{send1} and~\ref{allallAlg}.\ref{send0}, each inside a loop of $r'$ iterations, each nested in another loop of $\lceil \log n \rceil$ iterations (Line~\ref{allallAlg}.\ref{outerloop}), and the latter nested in a loop of $n'$ iterations (Line~\ref{allallAlg}.\ref{while}), where $n'$ is the number of different \rumors. Thus, the total number of rounds excluding \algfullname(\ ) and \mult(\ ) is $$2 r' \lceil \log n \rceil n' = \frac{2 \cT\ln n \lceil \log n \rceil n' }{ \ln(1+i_{\min})}.$$

Combining the latter with the running times $O( 
\ln n / \phi_{\min}^2
)$
proved in 
Theorem~\ref{thm:mult} 
and 
$O( 
n\ln n \log (n/\ell) /
(\ell \min\left\{  {\phi_{\min}}^2,\ln(1+i_{\min})\right\})
)$
in
Corollary~\ref{cor:time} 
(to Theorem~\ref{thm:many}), 
the total time complexity is
$$
O\left( 
\frac{n\ln n \log (n/\ell)}
{\ell \min\left\{  {\phi_{\min}}^2,\ln(1+i_{\min})\right\}}
+
\frac{ n'\ln n \log n}{ \ln(1+i_{\min})}
\right)
\ .
$$

Using the bound in Equation~\ref{condvsisop} to replace $\phi_{\min}$, we obtain the claimed time complexity.
\end{proof}
\section{Counting the Number of Copies of an \rumor\label{sec:mult-algorithm}}

In the main \allallfull Algorithm~\ref{allallAlg}, at the end of each epoch we count the number of nodes that have a discovered \rumor as its own \rumor. We do counting using the following procedure \mult (refer to Algorithm~\ref{multAlg}). It tries to locally balance potential truncated to $c\log d$ bits for a sufficiently long time, and once it (almost) stabilizes, returns it after scaling by the number of all nodes $n$ (recall that in the main algorithm \allallfull all nodes compute $n$ in the very beginning).

\begin{algorithm}[htbp]
\caption{\mult algorithm for each node. $\rumor$ is the \rumor initially held by this node and $new\ \rumor$ is the \rumor whose multiplicity has to be counted. $N$ is the set of neighbors of this node in the current round -- the node does not know them, but can send and receive short messages to/from them. Parameters $d$, $r''$, and $c$ are as defined in Theorem~\ref{thm:mult}. 
}
\label{multAlg}
\DontPrintSemicolon
    \leIf(\tcp*[f]{Assign initial potential}){$\rumor=new\ \rumor$}{$\Phi\gets 1$}{$\Phi\gets 0$}
			\For{$round=1$ to $r''$}{ 
				Send $\langle\Phi\rangle$ 
				and Receive $\langle\Phi_v\rangle, \forall v\in N$ \tcp*{To/from neighbors \dk{in $N$}}
					$\Phi \gets \Phi + \sum_{v\in N}\lfloor d^{c-1}\Phi_v\rfloor/d^c - |N|\lfloor d^{c-1}\Phi\rfloor/d^c$ \tcp*{Share potential \dk{truncated to $c \log d$ bits}} 
			}\label{gossipends}
	{\bf return} \dk{$\Phi \cdot n$} rounded to the closest integer (up or down)\;\label{returndelta}
\end{algorithm}

\subsection{Analysis of \mult}
\label{sec:mult}

We analyze the evolution of potentials in the \mult algorithm as a Markov chain on the $\cT$-connected evolving graph 
$\mathcal{G} = \{\mathcal{G}_j\}_{j=0}^{\infty}$,  where 
$\mathcal{G}_j = \{G_{j\cT+t}\}_{t=1}^{\cT}$, 
as defined in Section~\ref{sec:convtime}.
We adjust the potentials distribution after each $\cT$ rounds by the error produced by the truncation of potentials. 
In different nodes such error may delay the convergence to the stationary distribution. To upper bound the total time, we upper bound the delay (that is, the additional error with respect to the stationary distribution due to truncation). As a worst case we assume that the convergence may be delayed at all~nodes.

For the purpose of the analysis, we conceptually divide time into blocks of $\cT$ rounds. 
Let $b=r''/\cT$ be the number of blocks in \mult.
We denote the vector of potentials at the beginning of round $i$ of block $j$ as $\vec{\Phi}_{i,j}$. For clarity, we will sometimes refer to round $i$ of block $j$ as round $\langle i,j\rangle$. For inductive arguments we may refer to the round that follows (resp. preceeds) as $\langle i+1,j\rangle$ (resp. $\langle i-1,j\rangle$), omitting the fact that such round may correspond to a different block (i.e., $\langle 1,j+1\rangle$ and $\langle \cT,j-1\rangle$, resp.) -- this is only for notation and will not influence technical arguments.
Also, we denote the potential right after the potential distribution is stopped (i.e. after the loop in Line~\ref{multAlg}.\ref{gossipends}) as $\vec{\Phi}_{1,b+1}$, meaning the potential after the last block of rounds $b$.

\begin{theorem}
\label{thm:mult}
Consider an \ADCS 
with a $\cT$-connected evolving graph topology
with $n\geq 2$ nodes 
running the \mult algorithm 
with parameters
$d\geq 2n$,
$c \geq 5\alpha + 2\cT+4$,
and $r'' = \cT b$,
where
$b\geq 4 \alpha \ln n / \phi_{\min}^2$,
$\alpha \geq \max\{\log_n (3\cT),3\}$. 
Let $1\leq\delta\leq n$ be the number of copies being counted. 
Then, all nodes return the correct count $\delta$. 
\end{theorem}

\begin{proof}
The algorithm runs under the \ADCS restrictions because the potentials are truncated to $O(\log n)$ bits.

Consider the initial distribution $\vec{\Pi}_1$ on the overall potential $||\vec{\Phi}_{1,1}||_1 = \delta$. 
Given that $d\geq 2n$, the transition matrix is strongly aperiodic. Then, using Theorem~\ref{thm:mihail}, we know that after after $\cT$ rounds the distribution $\vec{\Pi}_{\cT+1}$ at the beginning of round $\cT+1$ would be such that
\begin{align}
\left|\left|\vec{\Pi}_{\cT+1} - \frac{\vec{I}}{n}\right|\right|_2^2 &\leq \left(1-\phi(P_0)^2\right) \left|\left|\vec{\Pi}_1 - \frac{\vec{I}}{n}\right|\right|_2^2 \ .
\label{rweqall}
\end{align}
Where $\phi(P_0)$ is the conductance of $P_0 = \prod_{t=1}^{\cT} P^{(t)}$, 
and $P^{(t)}$ is the matrix of shares used by the algorithm in round $t$. 

On the other hand, due to truncation, the vector of potentials is such that 
\begin{align*}
\left|\left|\frac{\vec{\Phi}_{1,2}}{||\vec{\Phi}_{1,1}||_1} - \frac{\vec{I}}{n}\right|\right|_2^2 
&\leq \sum_{v\in V} \left(\left|\Pi_{\cT+1}(v) - \frac{1}{n}\right| + \xi(v)\right)^2, 
\end{align*}
where $\xi(v)$ is the error introduced by the truncation at node $v$ during the $\cT$ rounds. As a worst case we have assumed that $\xi(v)$ contributes to the deviation with respect to the stationary distribution. Then, we have that
\begin{align}
\left|\left|\frac{\vec{\Phi}_{1,2}}{||\vec{\Phi}_{1,1}||_1} - \frac{\vec{I}}{n}\right|\right|_2^2 
&\leq \sum_{v\in V} \left(\left(\Pi_{\cT+1}(v) - \frac{1}{n}\right)^2 +2\xi(v)\left|\Pi_{\cT+1}(v) - \frac{1}{n}\right| + \xi(v)^2\right)
\nonumber\\ 
&=  
\left|\left|\vec{\Pi}_{\cT+1} - \frac{\vec{I}}{n}\right|\right|_2^2 + \sum_{v\in V} \left(2\xi(v)\left|\Pi_{\cT+1}(v) - \frac{1}{n}\right| + \xi(v)^2\right)
\nonumber\\
&\leq  
\left|\left|\vec{\Pi}_{\cT+1} - \frac{\vec{I}}{n}\right|\right|_2^2 
+ \sum_{v\in V} \left(2\xi(v) + \xi(v)^2\right).\label{boundeqall}
\end{align}

The potential received from each neighboring node is truncated to $c \log d$ bits after dividing by $d$. 
Thus, the error introduced in the calculation of its new potential in each round is at most $(n-1)/d^c \leq 1/d^{c-1}$.
Therefore, we have that after $\cT$ rounds it is $\xi(v)\leq \cT /d^{c-1}$, which 
given that $c \geq 5\alpha + 2\cT+4 \geq \log_d \cT + 1$ it is $\xi(v) \leq 1$. Then,
\begin{align}
\sum_{v\in V} \left(2\xi(v) + \xi(v)^2\right)
&\leq\sum_{v\in V} 3\xi(v) 
\leq\sum_{v\in V} \frac{3\cT }{d^{c-1}} 
=  \frac{3\cT n}{d^{c-1}}
\leq \frac{3\cT}{d^{c-2}}
\ .
\label{erroreqall}
\end{align}

Replacing Eqs.~\ref{rweqall} and~\ref{erroreqall} in~\ref{boundeqall}, we have that
\begin{align*}
\left|\left|\frac{\vec{\Phi}_{1,2}}{||\vec{\Phi}_{1,1}||_1} - \frac{\vec{I}}{n}\right|\right|_2^2 
&\leq \left(1-\phi(P_0)^2\right) \left|\left|\vec{\Pi}_1 - \frac{\vec{I}}{n}\right|\right|_2^2 
+ \frac{3\cT}{d^{c-2}}
\ \le \ 
\left(1-\phi(P_0)^2\right) \left|\left|\frac{\vec{\Phi}_{1,1}}{||\vec{\Phi}_{1,1}||_1} - \frac{\vec{I}}{n}\right|\right|_2^2 
+ \frac{3\cT}{d^{c-2}}
\ .
\end{align*}

The distribution of normalized potential at the beginning of round $\cT+1$ (which is the same as round $\langle 1,2\rangle$) is also a probability distribution.
Thus, the above analysis applies inductively to every subsequent block of $\cT$ rounds.

Then, the vector  of  potentials after $b$ blocks is such that
\begin{align*}
\left|\left|\frac{\vec{\Phi}_{1,b+1}}{||\vec{\Phi}_{1,1}||_1} - \frac{\vec{I}}{n}\right|\right|_2^2 
&\leq \left(1-\phi_{\min}^2\right)^b \left|\left|\frac{\vec{\Phi}_{1,1}}{||\vec{\Phi}_{1,1}||_1} - \frac{\vec{I}}{n}\right|\right|_2^2
+ \frac{3\cT}{d^{c-2}} 
\sum_{i=1}^{b-1}\left(1-\phi_{\min}^2\right)^i\\
&\leq 
\left(1-\phi_{\min}^2\right)^b \left|\left|\frac{\vec{\Phi}_{1,1}}{||\vec{\Phi}_{1,1}||_1} - \frac{\vec{I}}{n}\right|\right|_2^2
+ \frac{3\cT}{d^{c-2}}\frac{1}{\phi_{\min}^2}
\ \le \ 
\exp\left(-b\phi_{\min}^2\right)
+ \frac{3\cT}{d^{c-2}\phi_{\min}^2}
\ .
\end{align*}

In potential distribution algorithms where each non-zero entry of $P^{(j\cT+i)}$ is at least $1/d$, for any $j>0$, $i\in[1,\cT]$, and some $d\geq 1$, as it is the case in the \mult algorithm, it is $\phi_{\min} \geq 2/(nd^{\cT}) \geq 2/d^{\cT+1}$. Then, it is
\begin{align}
\left|\left|\frac{\vec{\Phi}_{1,b+1}}{||\vec{\Phi}_{1,1}||_1} - \frac{\vec{I}}{n}\right|\right|_2^2 
&\leq \exp\left(-b\phi_{\min}^2\right)
+ \frac{3\cT}{4d^{c-2\cT-4}}
\ .\label{eqfrommult}
\end{align}

Given that $b\geq 4 \alpha \ln n / \phi_{\min}^2$, it is 
\begin{align*}
\left|\left|\frac{\vec{\Phi}_{1,b+1}}{||\vec{\Phi}_{1,1}||_1} - \frac{\vec{I}}{n}\right|\right|_2^2 
&\leq \frac{1}{n^{4\alpha}}
+ \frac{3\cT}{4d^{c-2\cT-4}}
\ .
\end{align*}

Replacing $\alpha \geq \log_n (3\cT)$ and $d\geq n$, it is 
\begin{align*}
\left|\left|\frac{\vec{\Phi}_{1,b+1}}{||\vec{\Phi}_{1,1}||_1} - \frac{\vec{I}}{n}\right|\right|_2^2 
&\leq \frac{1}{n^{4\alpha}}
+ \frac{1}{4n^{c-2\cT-4-\alpha}}
\ .
\end{align*}

And for $c\geq 5\alpha + 2\cT+4$ we have that
\begin{align*}
\left|\left|\frac{\vec{\Phi}_{1,b+1}}{||\vec{\Phi}_{1,1}||_1} - \frac{\vec{I}}{n}\right|\right|_2^2 
&\leq \frac{5}{4n^{4\alpha}}
\leq \frac{1}{n^{2\alpha}}, \textrm{ for $\alpha\geq1/2$ and $n \geq 5/4$}
\ .
\end{align*}

Given that $({\Phi}_{1,b+1}[v]/||\vec{\Phi}_{1,1}||_1 - 1/n)^2 \leq \left|\left| \vec{\Phi}_{1,b+1}/||\vec{\Phi}_{1,1}||_1 - \frac{\vec{I}}{n}\right|\right|_2^2$ for any node $v$, we have that 
$({\Phi}_{1,b+1}[v]/||\vec{\Phi}_{1,1}||_1-1/n)^2 
\leq 1/n^{2\alpha}$ and hence $1/n - 1/n^{\alpha} \leq {\Phi}_{1,b+1}[v] /  ||\vec{\Phi}_{1,1}||_1 \leq 1/n + 1/n^{\alpha}$.

Then, we have
$(\delta/n - 1/n^{\alpha-1}) \leq {\Phi}_{1,b+1}[v] \leq (\delta/n + 1/n^{\alpha-1})$
and so
$\delta - 1/n^{\alpha-2} \leq {\Phi}_{1,b+1}[v] n \leq \delta + 1/n^{\alpha-2}$.
Replacing $\alpha\geq 3$ and $n\geq 2$
we have that
$\delta - 1/2 \leq {\Phi}_{1,b+1}[v] n \leq \delta + 1/2$.
Thus, after rounding to the closest integer in Line~\ref{multAlg}.\ref{returndelta} the returned value is $\delta$ as claimed. 
\end{proof}


\section{Counting the Number of Nodes in the System}
\label{sec:counting}

The first step of our \allallfull algorithm is to compute the unknown number of nodes $n$.
Counting the number of nodes in \ADNs has attracted a lot of attention recently~\cite{spirakis,KowalskiMicalp18journal,KowalskiMicalp19journal,conscious,LunaB15}. Most notably, the \MC (MC) protocol~\cite{KowalskiMicalp18journal} was the first one to achieve polynomial time after a flurry of papers improving bounds that started with doubly-exponential time~\cite{spirakis}.
The ADN model is challenging due to lack of node identifiers and arbitrary topology changes, but still lags behind real-world constraints such as limited bandwidth, disconnection of the network, and/or limited computational resources. 

In this work, as an integral part of our \allallfull algorithm, but also of independent interest, we present the \algfullname (\algname) algorithm which computes the number of nodes under the harsh restrictions of \adcs with opportunistic $\cT$-connectivity. 
\algname is inspired on an extension of MC called \MMC (MMC)~\cite{KowalskiMicalp19journal}, but carefully adapted to cope with limited messages, memory accesses, and connectivity. 
Our analysis relies on a bound on mixing time for $\cT$-connected evolving graphs. 
The analysis and design of the algorithm handle the errors produced by truncation of the calculations, purposely introduced to stay within the limits on message size and memory access.
We describe \algname in the following section and analyze~its~properties~afterwards.
The pseudocode of \algname can be found in Algorithms~\ref{leaderAlg} and~\ref{otherAlg}.

\subsection{The \algfullname Algorithm}
\label{algorithm}

\algname is executed synchronously by all nodes. Given that no preliminary knowledge of the network size is available, nodes maintain a running estimate of $n$ that starts with the minimum possible, that is, one more than the number of \ldrs $\ell$. By means of various alarms, the estimate is detected to be either low, high, or correct. Starting with the initial value, the estimate is updated by exponential search until it is correct. If the estimate becomes high (that is, if the correct value is skipped) the correct estimate is found by binary search in the last range between low and~high~detection. 

The evaluation of each estimate is implemented as follows. In each round of communication nodes share some potential values with neighbors in a gossip-based fashion. The fraction of potential shared, as well as the truncation of messages and calculations to adapt to the limits in communication and memory access, are functions of the running estimate, given that $n$ is not known. 
This gossiping continues for a number of rounds after which \ldrs move their potential to a separate accumulator completing what is called conceptually a phase. The process repeats for a number of phases to complete the evaluation of the current estimate.
The number of rounds and phases are functions of the running estimate, the errors produced by truncation, and the disconnection characteristic $\cT$.

If the estimate is correct, the number of rounds and phases are such that \ldrs have moved almost all the initial potential to their accumulators. Then, \ldrs univocally decide that the estimate is correct comparing the value in their accumulators with   some carefully designed range, and they disseminate the news to all other nodes by simple broadcast. 

If on the other hand the estimate is incorrect, that is, if it is either low or high, the algorithm detects by one or more of the following alarms. If nodes receive messages from more neighbors than the running estimate, the estimate is obviously low. If early in the execution potentials are above some carefully calculated threshold, the estimate is low as our analysis shows. Finally, if the \ldrs accumulated potential is outside the abovementioned range, or a node receives a notification of wrong estimate from other nodes, the estimate is either low or high (whatever corresponds).

A crucial difference between \algname and \nameD is the way that potentials are shared and updated in each round to cope with limits in communication and memory access. Specifically, in \nameD the potential $\Phi$ of each node $u$ is updated in round $r$ of phase $p$ of an epoch with estimate $k$ as
$$\Phi_{p,r+1}[u] = \Phi_{p,r}[u] + \sum_{v\in N_{p,r}[u]}\frac{\Phi_{p,r}[v]}{d} - |N_{p,r}[u]|\frac{\Phi_{p,r}[u]}{d},$$ 
where $N_{p,r}[u]$ is the set of neighbors of $u$ in round $r$ of phase $p$, $d=2k^{1+\epsilon}$ and $\epsilon>0$ is an arbitrarily small constant~\footnote{Notice that nodes do not know $N_{p,r}[u]$, but they know $|N_{p,r}[u]|$ after receiving messages. Also, recall that nodes are labeled only for the presentation, but nodes do not have identifiers.}.

In words, each node shares a fraction $1/d$ with its neighbors. So, node $u$ adds to its current potential a $1/d$ fraction of the potential of each neighbor, and subtracts a $1/d$ fraction of its own potential for each neighbor. 
The potential $\Phi$ is transmitted as a whole rather than the fraction since all nodes know $d=2k^{1+\epsilon}$ because the algorithm is synchronous. However, after successive rounds of dividing by $d$, the potential of some nodes may require $\omega(\log d)$ bits of precision. That is, for $k\in\omega(\sqrt[1+\epsilon]{n})$, they would require $\omega(\log n)$ bits of precision, violating the limits on communication and memory access. 
To avoid it, in \algname shared potentials are truncated to the most significant $c\log d$ bits, for some constant $c$ defined later. More precisely, in \algname the potential $\Phi$ of each node $u$ is updated in round $r$ of phase $p$ of an epoch with estimate $k$ as
\begin{align*}
\Phi_{p,r+1}[u] 
&= \Phi_{p,r}[u] + \sum_{v\in N_{p,r}[u]}\frac{\lfloor d^{c-1}\Phi_{p,r}[v]\rfloor}{d^c} - |N_{p,r}[u]|
\frac{\lfloor d^{c-1}\Phi_{p,r}[u]\rfloor}{d^c}\ .
\end{align*}
Also, to attain strong aperiodicity as required by Theorem~\ref{thm:mihail}, we set $d=2k^{1+\epsilon}$.
Given that throughout the execution $k<2n$, \algname does not violate the limits on communication and memory access.


\begin{algorithm}[htbp]
\caption{\algfullname algorithm for each {\bf\emph{\ldr node}}. $N$ is the set of neighbors of this node in the current round -- the node does not know them, but can send and receive short messages to/from them. $\ell$ is the number of \ldr nodes. The parameters $d,p,r,\tau$ and $c$ are as defined in Theorem~\ref{thm:many}.}
\label{leaderAlg}
\DontPrintSemicolon
	$k \gets \ell+1, min\gets k, max\gets\infty$ \tcp*{initial size estimate and range}
	\Repeat(\tcp*[f]{iterating epochs}){$status=done$}{  \label{epochsleader}
		$status\gets probing$ \tcp*{status$=$probing$|$low$|$high$|$done}
		$\Phi\gets 0$ \tcp*{current potential}
		$\rho\gets 0$ \tcp*{potential accumulator}
		\For(\tcp*[f]{iterating phases}){$phase=1$ to $p$}{  \label{phasesleader}
			\For( \tcp*[f]{iterating rounds}){$round=1$ to $r$}{ \label{roundsleader}
				Send $\langle\Phi,status\rangle$ 
				and Receive $\langle\Phi_i,status_i\rangle, \forall i\in N$ \tcp*{To/from neighbors \dk{in $N$}}
				\If{$status=probing$ {\bf and} $|N|< d/2$ {\bf and} $\forall i\in N:status_i=probing$} {
                    $\Phi \gets \Phi + \sum_{v\in N}\lfloor d^{c-1}\Phi_v\rfloor/d^c - |N|\lfloor d^{c-1}\Phi\rfloor/d^c$ \label{potupdate} \tcp*{Share potential \dk{truncated to $c \log d$ bits}}
				}
				\Else(\tcp*[f]{$k<n$}){  \label{leadertoomany}
					$status\gets low$, \label{alarminsecondleader}
					$\Phi\gets \ell$\;
				}
			}
			\lIf(\tcp*[f]{$k<n$}){$phase=1$ {\bf and} $\Phi> \tau$}{ 
					$status\gets low$,
					$\Phi\gets \ell$  \label{leaderthreshold}
			}
			\If(\tcp*[f]{prepare for next phase}){$status=probing$}{ 
				$\rho \gets \rho + \Phi$ \label{consume} \label{rhoupdate} \tcp*{consume potential} 
				$\Phi \gets 0$ \; \label{phireset}
			}
		}
		\If{$status=probing$}{ \label{statuscheck}
			\If(\tcp*[f]{$k=n$}){$(k-\ell)(1-k^{-\gamma})\leq \rho \leq (k-\ell)(1+k^{-\gamma})$}{ 
				$status\gets done$} \label{range}
			\lIf(\tcp*[f]{$k>n$}){$\rho < (k-\ell)(1-k^{-\gamma})$}{ 
				$status\gets high$}  \label{toobig}
			\lIf(\tcp*[f]{$k<n$}){$\rho > (k-\ell)(1+k^{-\gamma})$}{ 
				$status\gets low$}  \label{toolow}
		}
		\For(\tcp*[f]{disseminate status}){$round=1$ to $d$}{  
			Broadcast $\langle status\rangle$
			and Receive $\langle status_i\rangle, \forall i\in N$ \label{leadernotification}\;
		} 
		\eIf(\tcp*[f]{prepare for next epoch}){$status=low$}{  \label{leaderupdate}
			$min\gets k+1$ \;
			\leIf{$max=\infty$}{$k\gets 2k$}{$k\gets\lfloor(min+max)/2\rfloor$}
		}{
			\If{$status=high$}{
		         	$max\gets k-1$ \;
				$k\gets\lfloor(min+max)/2\rfloor$ \;
			}
		}
	}
	\textbf{return} $k$ \;

\end{algorithm}

\begin{algorithm}[htbp]
\caption{\algfullname algorithm for each {\bf\emph{\notldr node}}. $N$ is the set of neighbors of this node in the current round -- the node does not know them, but can send and receive short messages to/from them. $\ell$ is the number of \ldr nodes. The parameters $d,p,r,\tau$ and $c$ are as defined in Theorem~\ref{thm:many}.}
\label{otherAlg}
\DontPrintSemicolon
	$k \gets \ell+1, min\gets k, max\gets\infty$ \tcp*{initial size estimate and range}
	\Repeat(\tcp*[f]{iterating epochs}){$status=done$}{  \label{epochsother}
		$status\gets probing$ \tcp*{status$=$probing$|$low$|$high$|$done}
		$\Phi\gets \ell$ \tcp*{current potential}
		\For(\tcp*[f]{iterating phases}){$phase=1$ to $p$}{ \label{phasesother}
			\For(\tcp*[f]{iterating rounds}){$round=1$ to $r$}{  \label{roundsother}
				Send $\langle\Phi,status\rangle$
				and Receive $\langle\Phi_i,status_i\rangle, \forall i\in N$ \tcp*{To/from neighbors \dk{in $N$}}
				\If{$status=probing$ {\bf and} $|N|< d/2$ {\bf and} $\forall i\in N:status_i=probing$}{ 
					$\Phi \gets \Phi + \sum_{v\in N}\lfloor d^{c-1}\Phi_v\rfloor/d^c - |N|\lfloor d^{c-1}\Phi\rfloor/d^c$ \label{newpot} \tcp*{Share potential \dk{truncated to $c \log d$ bits}}
				}	
				\Else(\tcp*[f]{$k<n$}){  \label{othertoomany}
					$status\gets low$, \label{alarminsecondother}
					$\Phi\gets \ell$ \;
				}
			} 
			\lIf(\tcp*[f]{$k<n$}){$phase=1$ {\bf and} $\Phi> \tau$}{  
					$status\gets low$,
					$\Phi\gets \ell$ \label{otherthreshold}
			}
		}
		\For(\tcp*[f]{disseminate status}){$round=1$ to $d$}{  
			Broadcast $\langle status\rangle$
			and Receive $\langle status_i\rangle, \forall i\in N$ \label{othernotification}\;
			\lIf{$\exists i\in N:status_i\neq probing$}{$status\gets status_i$}
		}
		\If(\tcp*[f]{prepare for next epoch} ){$status=low$}{ \label{otherupdate}
			$min\gets k+1$ \;
			\leIf{$max=\infty$}{$k\gets 2k$}{$k\gets\lfloor(min+max)/2\rfloor$}
		}	
		\Else{ 
			\If{$status=high$}{
		         	$max\gets k-1$\;
				$k\gets\lfloor(min+max)/2\rfloor$\;
			}
		}
	}
	\textbf{return} $k$ \;

\end{algorithm}


\subsection{Analysis of \algfullname} 

We analyze the evolution of potentials as a Markov chain on the $\cT$-connected evolving graph 
$\mathcal{G} = \{\mathcal{G}_j\}_{j=0}^{\infty}$,  where 
$\mathcal{G}_j = \{G_{j\cT+t}\}_{t=1}^{\cT}$, 
as defined in Section~\ref{sec:convtime}.
We adjust the potentials distribution after each $\cT$ rounds by the error produced by the truncation of potentials. 
In different nodes such error may delay the convergence to the stationary distribution. To upper bound the total time, we upper bound the delay (that is, the additional error with respect to the stationary due to truncation). As a worst case we assume that the convergence is delayed at all nodes.

For the purpose of the analysis, we conceptually divide each phase in blocks of $\cT$ rounds. (The number of rounds in each phase will be a multiple of $\cT$.) We denote the vector of potentials at the beginning of round $i$ of block $j$ of phase $h$ as $\vec{\Phi}_{i,j,h}$. For clarity, we will sometimes refer to round $i$ of block $j$ of phase $h$ as round $\langle i,j,h\rangle$, and for inductive arguments we may refer to the round that follows (resp. preceeds) as $\langle i+1,j,h\rangle$ (resp. $\langle i-1,j,h\rangle$), omitting the fact that such round may correspond to a different block.
Also, we refer to the potential right before the beginning of phase $h>1$, denoted as $\vec{\Phi}_{0,0,h}$, meaning the potential after the last round of phase $h-1$, before the potential is reset at the beginning of phase $h$.


First, we prove the following two claims 
about
properties of the potential during the execution of \algname, 
for later use.

\begin{claim}
\label{manyconservation}

Given an \ADCS of $n$ nodes running \algname with parameter $d$, for any round $\langle i,j,h\rangle$ of phase $h$, if $d$ was larger than the number of neighbors of each node $x$ for every round $\langle i',j',h\rangle$ before round $\langle i,j,h\rangle$, then $  ||\vec{\Phi}_{i,j,h}||_1 = ||\vec{\Phi}_{1,1,h}||_1$. 
\end{claim} 

\begin{proof}
For any given round $\langle i,j,h\rangle$ and any given node $u$, if 
$d/2$ is larger than the number of neighbors of $u$, the potential is updated as (refer to Lines~\ref{leaderAlg}.\ref{potupdate} and~\ref{otherAlg}.\ref{newpot})
\begin{align*}
\Phi_{i+1,j,h}[u] 
&= \Phi_{i,j,h}[u] + \sum_{v\in N_{i,j,h}[u]}\frac{\lfloor d^{c-1}\Phi_{i,j,h}[v]\rfloor}{d^c} - |N_{i,j,h}[u]|
\frac{\lfloor d^{c-1}\Phi_{i,j,h}[u]\rfloor}{d^c}
\ ,
\end{align*} 
where $N_{i,j,h}[u]$ is the set of neighbors of $u$ in round $\langle i,j,h\rangle$. 
Inductively, assume that the claimed overall potential holds for some round $\langle i,j,h\rangle$, 
we want to show that consequently it holds for $\langle i+1,j,h\rangle$.
The potential for round $\langle i+1,j,h\rangle$ is
\begin{align}
||\vec{\Phi}_{i+1,j,h}||_1 &= ||\vec{\Phi}_{i,j,h}||_1 + \frac{1}{d^c}\sum_{u\in V} \left( \sum_{v\in N_{i,j,h}[u]}\lfloor d^{c-1}\Phi_{i,j,h}[v]\rfloor - |N_{i,j,h}[u]|\lfloor d^{c-1}\Phi_{i,j,h}[u]\rfloor \right) \ .\label{manypotvecupdate}
\end{align}

In the \ADCS model, communication is symmetric. That is, for every pair of nodes $x,y\in V$, and round $\langle i,j,h \rangle$ it is $x\in N_{i,j,h}[y] \iff y\in N_{i,j,h}[x]$. 
Fix a pair of nodes $x',y' \in V$ such that in round $\langle i,j,h\rangle$ it is $y'\in N_{i,j,h}[x']$ and hence $x'\in N_{i,j,h}[y']$. 
Consider the summations in Equation~\ref{manypotvecupdate}.
Due to symmetric communication, we have that the potential $\Phi_{i,j,h}[y']$ appears with positive sign when the indices of the summations are $x=x'$ and $y=y'$, and with negative sign when the indices are $x=y'$ and $y=x'$. (Notice that it is truncated the same way at both nodes.) This observation applies to all pairs of nodes that communicate in any round. 
Therefore, we can re-write Equation~\ref{manypotvecupdate} as
\begin{align*}
||\vec{\Phi}_{i+1,j,h}||_1 &= ||\vec{\Phi}_{i,j,h}||_1 + \frac{1}{d^c}
\!\!
\sum_{\substack{u,v\in V:\\v\in N_{i,j,h}[u]\\}} \bigg(
\lfloor d^{c-1}\Phi_{i,j,h}[u]\rfloor - \lfloor d^{c-1}\Phi_{i,j,h}[v]\rfloor
+
\lfloor d^{c-1}\Phi_{i,j,h}[v]\rfloor - \lfloor d^{c-1}\Phi_{i,j,h}[u]\rfloor
\bigg) 
 = 
||\vec{\Phi}_{i,j,h}||_1 \ .
\end{align*}
Thus, the claim follows.
\end{proof}


\begin{claim}
\label{manypotbounds}
Given an \ADCS with $\ell>0$ \ldrs and $n-\ell>0$ \notldr nodes running \algname with parameter $c\geq 2$, for any round $\langle i,j,h\rangle$ and any node $u\in V$, it is $\Phi_{i,j,h}[u]\geq 0$.  
\end{claim}

\begin{proof}
At the beginning of the first round the potential of the \ldr nodes is $0$ and the potential of any \notldr node $x$ is $\ell$. Thus, the claim follows. 

Inductively, for any round $\langle i,j,h\rangle$ after the first round, we consider two cases according to node status. 
If a node $x$ is in alarm status ``low'' at the beginning of the round, then it is $\Phi_{i,j,h}[x]=\ell$ because, whenever the status of a node is updated to ``low'', its potential is set to $\ell$ and will not change until the next epoch (refer to Algorithms~\ref{leaderAlg} and~\ref{otherAlg})

In the second case, if a node $u$ is in ``probing'' status at the beginning of round $\langle i,j,h\rangle$, it means that it had its potential updated in all previous rounds $\langle i',j',h'\rangle$ as (refer to Lines~\ref{leaderAlg}.\ref{potupdate} or~\ref{otherAlg}.\ref{newpot})
\begin{align}
\Phi_{i'+1,j',h'}[u] 
&= \Phi_{i',j',h'}[u] + \sum_{v\in N_{i',j',h'}[u]}\frac{\lfloor d^{c-1}\Phi_{i',j',h'}[v]\rfloor}{d^c} - |N_{i',j',h'}[u]|
\frac{\lfloor d^{c-1}\Phi_{i',j',h'}[u]\rfloor}{d^c}\nonumber\\
&\leq \Phi_{i',j',h'}[u] + \sum_{v\in N_{i',j',h'}[u]}\frac{\Phi_{i',j',h'}[v]}{d} - |N_{i',j',h'}[u]|
\frac{\Phi_{i',j',h'}[u]}{d} 
+ |N_{i',j',h'}[u]|\frac{1}{d^c}\ .\label{updateeq}
\end{align}

Where the last term is the maximum extra potential kept by node $u$ due to truncation to $c \log d$ bits by nodes in $N_{i,j,h}[u]$.
For all rounds $\langle i',j',h'\rangle$, node $u$ exchanged potential with less than 
$d/2$ neighbors, because otherwise it would have been changed to alarm status (refer to Lines~\ref{leaderAlg}.\ref{alarminsecondleader} and~\ref{otherAlg}.\ref{alarminsecondother}).
Therefore it is $|N_{i',j',h'}[u]|\Phi_{ i',j',h'}[u]/d < \Phi_{i',j',h'}[u]$ which implies $\Phi_{i,j,h}[u]\geq 0$. 
\end{proof}


The general structure of the rest of the proof is the following. We divide the analysis in 4 cases according to the relation between the estimate $k$ and the network size $n$.
\vspace{0.1in}
\begin{center}
\begin{tabular}{|c|c|c|}
\hline
if & then & proved in \\
\hline
$k^{1+\epsilon} < n$ & All nodes receive alarm & Lemma~\ref{manyalarmsoon} \& auxiliary~\ref{manyunalarmed} and~\ref{nolowalarm}.\\
\hline
$k < n \leq k^{1+\epsilon}$ & $\rho > (k-\ell) \left(1 + \frac{1}{k^\gamma}\right)$ & Lemma~\ref{manyksquare}.\\
\hline
$k = n$ & $(k-\ell) \left(1 - \frac{1}{k^\gamma}\right) \leq \rho \leq (k-\ell) \left(1 + \frac{1}{k^\gamma}\right)$ & Lemma~\ref{manycorrect}.\\
\hline
$k>n$ & $\rho <  (k-\ell)\left(1-\frac{1}{k^\gamma}\right)$ & Lemma~\ref{kaboven}.\\
\hline
\end{tabular}
\end{center}
\vspace{0.1in}

We begin the analysis considering the case $k=n$, as follows.

\begin{lemma}
\label{manycorrect}
Consider an \ADCS 
with a $\cT$-connected evolving graph topology
with $\ell>0$ \ldrs and $n-\ell>0$ \notldr nodes 
running the \algname protocol 
with parameters
$d\geq 2k$,
$p \geq (2\gamma\ln k)/\left(\ell\left(\frac{1}{k}+\frac{1}{k^{\alpha}}\right)\right)$, 
$r\geq 4\cT\alpha \ln k / \phi_{\min}^2$,
and $c\geq 5\alpha + 2\cT + 4$,
where
$\gamma>0$ 
and
$\alpha \geq \max\{2,1+\gamma+\log_k 3,\log_k (3\cT)\}$.
Then, if 
$k = n$, 
the potential $\rho$ consumed by each of the $\ell$ \ldr nodes is such that 
$$(k-\ell) \left(1 - \frac{1}{k^\gamma}\right) \leq \rho \leq (k-\ell) \left(1 + \frac{1}{k^\gamma}\right).$$
\end{lemma}

\begin{proof}

Given that $d\geq 2n$ and $c > \log_d \cT + 1$, we have from Eq.~\ref{eqfrommult} in the proof of Theorem~\ref{thm:mult} (changing the notation appropriately) that right before the beginning of any phase $h>1$ it is 

\begin{align*}
\left|\left|\frac{\vec{\Phi}_{0,0,2}}{||\vec{\Phi}_{1,1,1}||_1} - \frac{\vec{I}}{n}\right|\right|_2^2 
&\leq \exp\left(-b\phi_{\min}^2\right)
+ \frac{3\cT}{4d^{c-2\cT-4}}
\ .
\end{align*}

The above applies to any phase as long as $d\geq n$. Thus, we have that right before the beginning of any phase $h>1$:
\begin{align}
\left|\left|\frac{\vec{\Phi}_{0,0,h}}{||\vec{\Phi}_{1,1,h-1}||_1} - \frac{\vec{I}}{n}\right|\right|_2^2 
&\leq \exp\left(-b\phi_{\min}^2\right)
+ \frac{3\cT}{4d^{c-2\cT-4}}
\ .
\label{disteq}
\end{align}

Given that $r\geq 4\cT \alpha \ln k / \phi_{\min}^2$ and each block has $\cT$ rounds, it is $b\geq 4\alpha \ln k / \phi_{\min}^2$. Thus, 
\begin{align*}
\left|\left|\frac{\vec{\Phi}_{0,0,h}}{||\vec{\Phi}_{1,1,h-1}||_1} - \frac{\vec{I}}{n}\right|\right|_2^2 
&\leq \frac{1}{k^{4\alpha}}
+ \frac{3\cT}{4d^{c-2\cT-4}}
\ .
\end{align*}

Replacing $\alpha \geq \log_k (3\cT)$ and $d\geq k$, it is 
\begin{align*}
\left|\left|\frac{\vec{\Phi}_{0,0,h}}{||\vec{\Phi}_{1,1,h-1}||_1} - \frac{\vec{I}}{n}\right|\right|_2^2 
&\leq \frac{1}{k^{4\alpha}}
+ \frac{1}{4d^{c-2\cT-4-\alpha}}
\ .
\end{align*}

And for $c\geq 5\alpha + 2\cT+4$ we have that
\begin{align*}
\left|\left|\frac{\vec{\Phi}_{0,0,h}}{||\vec{\Phi}_{1,1,h-1}||_1} - \frac{\vec{I}}{n}\right|\right|_2^2 
&\leq \frac{5}{4k^{4\alpha}}
\leq \frac{1}{k^{2\alpha}}\ , \textrm{ for $\alpha\geq1/2$ and $k \geq 2$}
\ .
\end{align*}

Given that $({\Phi}_{0,0,h}[v]/||\vec{\Phi}_{1,1,h-1}||_1 - 1/n)^2 \leq \left|\left| \vec{\Phi}_{0,0,h}/||\vec{\Phi}_{1,1,h-1}||_1 - \frac{\vec{I}}{n}\right|\right|_2^2$ for any node $v$ and phase $h>1$, we have that $
({\Phi}_{0,0,,h}[v]/||\vec{\Phi}_{1,1,h-1}||_1-1/n)^2 
=({\Phi}_{0,0,,h}[v]/||\vec{\Phi}_{1,1,h-1}||_1-1/k)^2 
\leq 1/k^{2\alpha}$ and hence ${\Phi}_{0,0,h}[v] \geq (1/k - 1/k^{\alpha}) ||\vec{\Phi}_{1,1,h-1}||_1$.
Notice that the latter is true for any initial distribution.
Therefore, after each phase a \ldr node consumes 
between $1/k-1/k^{\alpha}$ and $1/k+1/k^{\alpha}$ fraction of the total
potential in the system, and the total potential in the system drops by at least
$\ell(1/k-1/k^{\alpha})$  and by at most $\ell(1/k+1/k^{\alpha})$ fraction. 
Recall that the initial overall potential in the system is $\ell(n-\ell)=\ell(k-\ell)$, and that by Claim~\ref{manyconservation}, if $d>n$, the overall potential in the system is the same throughout each phase.

Using the latter observations, we first find conditions on the number of phases $p$ to obtain the desired bounds on $\rho$, as follows.
After $p$ phases a \ldr node consumes {\em at least}
\begin{align}
\rho &\geq \ell(k-\ell) \left(\frac{1}{k} - \frac{1}{k^{\alpha}}\right) \sum_{i=0}^{p-1} \left(1- \ell\left(\frac{1}{k}+\frac{1}{k^{\alpha}}\right)\right)^i \ ,\label{rholb}
\end{align}
and {\em at most}
\begin{align}
\rho &\leq \ell(k-\ell) \left(\frac{1}{k} + \frac{1}{k^{\alpha}}\right) \sum_{i=0}^{p-1} \left(1- \ell\left(\frac{1}{k}-\frac{1}{k^{\alpha}}\right)\right)^i
\ .
\label{rhoub}
\end{align}
Given that $0<\ell\left(\frac{1}{k}+\frac{1}{k^{\alpha}}\right)<1$ for $\alpha \geq 2$ and $k>\ell$, Equation~\ref{rholb} is
\begin{align*}
\rho &\geq \ell(k-\ell) \left(\frac{1}{k} - \frac{1}{k^{\alpha}}\right) \frac{1-\left(1- \ell\left(\frac{1}{k}+\frac{1}{k^{\alpha}}\right)\right)^p}{1-\left(1- \ell\left(\frac{1}{k}+\frac{1}{k^{\alpha}}\right)\right)} 
\ = \ 
(k-\ell) 
\frac{ k^{\alpha}-k}{k^{\alpha}+k}
\left(1-\left(1- \ell\left(\frac{1}{k}+\frac{1}{k^{\alpha}}\right)\right)^p\right)
\ .
\end{align*}
Given that $0<\ell\left(\frac{1}{k}+\frac{1}{k^{\alpha}}\right)<1$ 
for $\ell<k$ and $\alpha \geq 2$, it is 
\begin{align*}
\rho &\geq (k-\ell) 
\frac{ k^{\alpha}-k}{k^{\alpha}+k}
\left(1-\exp\left(- p\ell\left(\frac{1}{k}+\frac{1}{k^{\alpha}}\right)\right)\right)
\ .
\end{align*}
Thus, to prove the lower bound on $\rho$, it is enough to find values of $p$ and $\alpha$ such that
\begin{align*}
\frac{ k^{\alpha}-k}{k^{\alpha}+k}
\left(1-\exp\left(- p\ell\left(\frac{1}{k}+\frac{1}{k^{\alpha}}\right)\right)\right)
&\geq 1-\frac{1}{k^\gamma} \ .
\end{align*}
We note first that for $$p\geq \frac{2\gamma\ln k}{\ell\left(\frac{1}{k}+\frac{1}{k^{\alpha}}\right)} \ ,$$ 
it is $$1-\exp\left(- p\ell\left(\frac{1}{k}+\frac{1}{k^{\alpha}}\right)\right) \geq 1-\frac{1}{k^{2\gamma}} \ .$$
Replacing, it is enough to prove
\begin{align}
\frac{ k^{\alpha}-k}{k^{\alpha}+k} \left(1+\frac{1}{k^\gamma}\right) &\geq 1 \label{eqceps}\\
k^{\alpha-\gamma} &\geq 2k + k^{1-\gamma} \ . \nonumber
\end{align}
Thus, for $\gamma>0$ it is enough to prove
$k^{\alpha-\gamma} \geq 3k$,
which is true for $\alpha \geq 1+\gamma+\log_k 3$. 
We show now the upper bound on $\rho$ starting from Equation~\ref{rhoub}:
\begin{align*}
\rho &\leq \ell(k-\ell) \left(\frac{1}{k} + \frac{1}{k^{\alpha}}\right) \sum_{i=0}^{p-1} \left(1- \ell\left(\frac{1}{k}-\frac{1}{k^{\alpha}}\right)\right)^i.
\end{align*}
Given that $1- \ell\left(\frac{1}{k}-\frac{1}{k^{\alpha}}\right)<1$ for $\alpha \geq 2>1$, it is
\begin{align*}
\rho &\leq \ell(k-\ell) \left(\frac{1}{k} + \frac{1}{k^{\alpha}}\right) \frac{1-\left(1- \ell\left(\frac{1}{k}-\frac{1}{k^{\alpha}}\right)\right)^p}{1-\left(1- \ell\left(\frac{1}{k}-\frac{1}{k^{\alpha}}\right)\right)} 
\ = \
(k-\ell) 
\frac{ k^{\alpha}+k}{k^{\alpha}-k}
\left(1-\left(1- \ell\left(\frac{1}{k}-\frac{1}{k^{\alpha}}\right)\right)^p\right).
\end{align*}
Given that $0<\ell\left(\frac{1}{k}-\frac{1}{k^{\alpha}}\right)<1$ and $p>0$, it is $\left(1-\left(1- \ell\left(\frac{1}{k}-\frac{1}{k^{\alpha}}\right)\right)^p\right)<1$. Then, replacing, we get
\begin{align*}
\rho &\leq (k-\ell) 
\frac{ k^{\alpha}+k}{k^{\alpha}-k} \ .
\end{align*}
Thus, to prove the upper bound on $\rho$, it is enough to show that
\begin{align*}
\frac{ k^{\alpha}+k}{k^{\alpha}-k}
&\leq 1+\frac{1}{k^\gamma} \ .
\end{align*}
This is the same as Equation~\ref{eqceps} and hence the claim follows.
\end{proof}


The previous lemma shows that, after running \algname enough time, if 
for some \ldr node it is $ \rho > (k-\ell) \left(1 + \frac{1}{k^\gamma}\right)$
or
$\rho < (k-\ell) \left(1 - \frac{1}{k^\gamma}\right)$,
for some $\gamma>0$,
we know that the estimate $k$ is wrong. However, the complementary case, that is, 
$(k-\ell) \left(1 - \frac{1}{k^\gamma}\right) \leq \rho \leq (k-\ell) \left(1 + \frac{1}{k^\gamma}\right)$,
may occur even if the estimate is $k\neq n$ and hence the error has to be detected by other means. 
To prove correctness in that case we further separate the range of $k$ in three cases. 
The first one, when $k<n\leq k^{1+\epsilon}$, for some $\epsilon>0$, in the following lemma, 
which is based on upper bounding the potential left in the system after running \algname long enough.

\begin{lemma}
\label{manyksquare}
Consider an \ADCS 
with a $\cT$-connected evolving graph topology
with $\ell>0$ \ldrs and $n-\ell>0$ \notldr nodes 
running the \algname protocol 
with parameters
$d\geq k^{1+\epsilon}$,
$p\geq 2\delta(\ln k)/(\ell\left(1/n+1/k^{\beta}\right))$,
$r\geq 4\cT \beta \ln k / \phi_{\min}^2$,
and $c\geq 5\beta+2\cT+4$, 
where
$\beta \geq \log_k \max\{(n(2k^{\delta} + 1)),3\cT\}$,
$\beta >2$,
$\delta > \log_k (nk^\gamma/(nk^\gamma-(n-1)(k^\gamma+1)))$, 
$\gamma >\log_k (n-1)$,
and $\epsilon>0$.
Then, if
$k < n \leq k^{1+\epsilon}$, 
the potential $\rho$ consumed by any \ldr node is $\rho >  (k-\ell)\left(1+1/k^\gamma\right)$.
\end{lemma}

\begin{proof}
Given that $d\geq n$ and $c > \log_d \cT + 1$, we have from Eq.~\ref{disteq} that right before the beginning of any phase $h>1$ 
\begin{align*}
\left|\left|\frac{\vec{\Phi}_{0,0,h}}{||\vec{\Phi}_{1,1,h-1}||_1} - \frac{\vec{I}}{n}\right|\right|_2^2 
&\leq \exp\left(-b\phi_{\min}^2\right)
+ \frac{3\cT}{4d^{c-2\cT-4}}
\ .
\end{align*}

Given that $r\geq 4\cT \beta \ln k / \phi_{\min}^2$ and each block has $\cT$ rounds, it is $b\geq 4\beta \ln k / \phi_{\min}^2$. Thus, 
\begin{align*}
\left|\left|\frac{\vec{\Phi}_{0,0,h}}{||\vec{\Phi}_{1,1,h-1}||_1} - \frac{\vec{I}}{n}\right|\right|_2^2 
&\leq \exp\left(-4\beta  \ln k\right)
+ \frac{3\cT}{4d^{c-2\cT-4}}
\ = \ 
\frac{1}{k^{4\beta}}
+ \frac{3\cT}{4d^{c-2\cT-4}}
\ .
\end{align*}

Given that $c\geq 5\beta+2\cT+4$ and $d>k$ we have that
\begin{align*}
\left|\left|\frac{\vec{\Phi}_{0,0,h}}{||\vec{\Phi}_{1,1,h-1}||_1} - \frac{\vec{I}}{n}\right|\right|_2^2 
&\leq \frac{1}{k^{4\beta}}
+ \frac{3\cT}{4d^{5\beta}}
\ .
\end{align*}

And for $d\geq k$ and $\beta\geq \log_k(3\cT)$, it is
\begin{align*}
\left|\left|\frac{\vec{\Phi}_{0,0,h}}{||\vec{\Phi}_{1,1,h-1}||_1} - \frac{\vec{I}}{n}\right|\right|_2^2 
&\leq \frac{5}{4k^{4\beta}}
\leq \frac{1}{k^{2\beta}}\ , \textrm{ for $\beta\geq1/2$ and $k \geq 2$} \ .
\end{align*}

Given that $(\Phi_{0,0,h}[v]/||\vec{\Phi}_{1,1,h-1}||_1 - 1/n)^2 \leq \left|\left|\frac{\vec{\Phi}_{0,0,h}}{||\vec{\Phi}_{1,1,h-1}||_1} - \frac{\vec{I}}{n}\right|\right|_2^2$ for any node $v$, we have that $(\Phi_{0,0,h}[v]/||\vec{\Phi}_{1,1,h-1}||_1-1/n)^2 \leq 1/k^{2\beta}$ and hence $\Phi_{0,0,h}[v] \geq (1/n - 1/k^{\beta}) ||\vec{\Phi}_{1,1,h-1}||_1$.
The latter is true for any initial distribution and any phase $h$.
Therefore, after each phase a \ldr node consumes 
at least $1/n-1/k^{\beta}$ fraction of the total
potential in the system, and the total potential in the system drops by at most $\ell(1/n+1/k^{\beta})$ fraction. 
Recall that the initial overall potential in the system is $\ell(n-\ell)$, and that by Claim~\ref{manyconservation}, if $d>n$, the overall potential in the system does not change during each phase.

Using the latter observations, after $p$ phases, any given \ldr node consumes {\em at least}
\begin{align*}
\rho &\geq \ell(n-\ell) \left(\frac{1}{n} - \frac{1}{k^{\beta}}\right) \sum_{i=0}^{p-1} \left(1- \ell\left(\frac{1}{n}+\frac{1}{k^{\beta}}\right)\right)^i.
\end{align*}
Given that $0<\ell\left(\frac{1}{n}+\frac{1}{k^{\beta}}\right)<1$ for $\beta\geq 2$ and $k>\ell$, we have that
\begin{align*}
\rho &\geq \ell(n-\ell) \left(\frac{1}{n} - \frac{1}{k^{\beta}}\right) \frac{1-\left(1- \ell\left(\frac{1}{n}+\frac{1}{k^{\beta}}\right)\right)^p}{1-\left(1- \ell\left(\frac{1}{n}+\frac{1}{k^{\beta}}\right)\right)} 
\ = \ 
(n-\ell)  \frac{ k^{\beta}-n}{k^{\beta}+n}
\left(1-\left(1- \ell\left(\frac{1}{n}+\frac{1}{k^{\beta}}\right)\right)^p\right)
\ .
\end{align*}
Again using that $0<\ell\left(\frac{1}{n}+\frac{1}{k^{\beta}}\right)<1$ for $\beta\geq 2$ and $k>\ell$, 
and given that $1-x\leq e^{-x}$ for any $0<x<1$~\cite{book:mitrinovic},
we have 
\begin{align*}
\rho &\geq (n-\ell) \frac{ k^{\beta}-n}{k^{\beta}+n}
\left(1-\exp\left(- p\ell\left(\frac{1}{n}+\frac{1}{k^{\beta}}\right)\right)\right) \ ,
\textrm{ replacing } p\geq \frac{2\delta\ln k}{\ell\left(\frac{1}{n}+\frac{1}{k^{\beta}}\right)} \ ,\\ 
&\geq (n-\ell) \frac{ k^{\beta}-n}{k^{\beta}+n} \left(1- \frac{1}{k^{2\delta}}\right)
\ \ge \ 
(n-\ell) \frac{ k^{\beta}-n}{k^{\beta}+n} \left(1+ \frac{1}{k^{\delta}}\right) \left(1- \frac{1}{k^{\delta}}\right)
\ \ge \ 
(n-\ell) \left(1- \frac{1}{k^{\delta}}\right) \ .
\end{align*}
The latter inequality holds for $\beta \geq \log_k (n(2k^{\delta} + 1))$.
Then, to complete the proof, it is enough to show that
\begin{align*}
(n-\ell) \left(1- \frac{1}{k^{\delta}}\right) &> (k-\ell)\left(1+\frac{1}{k^\gamma}\right)
\ .
\end{align*}
Which is true for $k<n$, $\delta> \log_k (nk^\gamma/(nk^\gamma-(n-1)(k^\gamma+1)))$ and $\gamma>\log_k (n-1)$.
\end{proof}


We now consider the case $k^{1+\epsilon}<n$.
We focus on the first phase. 
We define a threshold $\tau$ and a number of rounds $r$ such that, after the phase is completed, all nodes that have potential above $\tau$ can send an alarm to the leader, as such potential indicates that the estimate is low. 

In order to do that, we first establish an upper bound of at most $k^{1+\epsilon}$ nodes with potential at most $\tau$ at the end of the first phase (Lemma~\ref{manyunalarmed}).
Given that $k^{1+\epsilon}<n$, using this lemma we know that there is at least one node with potential above $\tau$ at the end of the first phase.
Second, we show that if the estimate is not low, that is $k\geq n$, then all nodes have potential at most $\tau$ at the end of the first phase (Lemma~\ref{nolowalarm}).  
That is, a potential above $\tau$ can only happen when indeed the estimate is low.
Finally, we show that if $k^{1+\epsilon}<n$ an alarm ``low'' initiated by nodes with potential above $\tau$ must be received after $k^{1+\epsilon}$ further rounds of communication (Lemma~\ref{manyalarmsoon}).

\begin{lemma}
\label{manyunalarmed}
Consider an \ADCS 
with a $\cT$-connected evolving graph topology
with $\ell>0$ \ldrs and $n-\ell>0$ \notldr nodes 
running the \algname protocol.
For $\epsilon>0$,
after running the first phase, 
there are at most $k^{1+\epsilon}$ nodes that 
have potential at most $\tau=\ell(1-\ell/k^{1+\epsilon})$.
\end{lemma}

\begin{proof}
We define the \emph{slack} of node $x$ at the beginning of round $\langle i,j,h\rangle$ as $s_{i,j,h}[x]=\ell-\Phi_{i,j,h}[x]$ and the vector of slacks at the beginning of round $\langle i,j,h\rangle$ as $\vec{s}_{i,j,h}$. In words, the slack of a node is the ``room'' for additional potential up to $\ell$. 
Recall that the overall potential at the beginning of round $\langle 1,1,1 \rangle$ is $||\vec{\Phi}_{1,1,1}||_1=(n-\ell)\ell$. 
Also notice that for any round and any node $x$ the potential of $x$ is non-negative as shown in Claim~\ref{manypotbounds}.
Therefore, the overall slack 
at the beginning of round $\langle1,1,1\rangle$ 
with respect to 
$\ell$ 
is $||\vec{s}_{1,1,1}||_1=\ell^2$.

Consider a partition of the set of nodes $\{L,H\}$, where $L$ is the set of nodes with potential at most $\tau$ at the end of the first phase, before the \ldr nodes consume their own potential in Line~\ref{leaderAlg}.\ref{consume}. That is, $\Phi_{0,0,2}[x] \leq \tau$ for all $x\in L$ (and $\Phi_{0,0,2}[y] > \tau$ for all $y\in H$).
Assume that the slack held by nodes in $L$ at the end of the first phase is at most the overall slack at the beginning of the phase. That is, $\sum_{x\in L}s_{0,0,2}[x] \leq \ell^2$. 
By definition of $L$, we have that for each node $x\in L$ it is $s_{0,0,2}[x] = \ell-\Phi_{0,0,2}[x] \geq \ell-\tau$.
Therefore,
$|L|(\ell-\tau) \leq \sum_{x\in L} s_{0,0,2}[x] \leq \ell^2$.
Thus, $|L| \leq \ell^2/(\ell-\tau) = k^{1+\epsilon}$ (because $\tau=\ell(1-\ell/k^{1+\epsilon})$) and the claim follows.

Then, to complete the proof, it remains to show that $\sum_{x\in L}s_{0,0,2}[x]\leq \ell^2$.
Let the scenario where $d/2$ is larger than the number of neighbors that each node has in each round of the first phase be called ``case 1'', and ``case 2'' otherwise.
Claim~\ref{manyconservation} shows that in case 1 at the end of the first phase it is $||\vec{\Phi}_{0,0,2}||_1=||\vec{\Phi}_{1,1,1}||_1=\ell(n-\ell)$. Therefore, the slack held by all nodes is $||\vec{s}_{0,0,2}||_1=\ell^2$ and thus the slack held by nodes in $L\subseteq V$ is $\sum_{x\in L}s_{0,0,2}[x]\leq \ell^2$, proving the claim for case 1. We show now that, in fact, case 1 is a worst-case scenario. That is, in the complementary case 2 where some nodes have $d/2$ neighbors or more in one or more rounds, the slack is even smaller. To compare both scenarios we denote the slack for each round $\langle i,j,h\rangle$, each node $x$, and each case $i$ as $s^{(i)}_{i,j,h}[x]$.

Assume that some node $x\in L$ is the first one to have 
$d'\geq d/2$ neighbors. Let 
round $\langle i,j,1\rangle$
be the first one in the first phase when this event happened.
We claim that 
$\sum_{x\in L}s^{(2)}_{i+1,j,1}[x] \leq 
\sum_{x\in L}s^{(1)}_{i+1,j,1}[x]$.
The reason is the following.
Given that at least $d/2$ potentials are received, node $x$ changes its potential to $\ell$ for the rest of the epoch (refer to Lines~\ref{leaderAlg}.\ref{leadertoomany} and~\ref{otherAlg}.\ref{othertoomany}). 
And given that $x\in L$, its potential before changing was at most $\tau<\ell$.
That is, the slack of $x$ is $s_{i+1,j,1}^{(2)}[x]\leq s_{i,j,1}^{(2)}[x]=s_{i,j,1}^{(1)}[x]$. 
Additionally, given that $d'\geq d/2$ the potential shared by $x$ with \emph{all} neighbors during round $\langle i,j,1\rangle$ in case 2 is not less than the potential that $x$ would have shared in case 1 when $|N|<d/2$ (refer to Lines~\ref{leaderAlg}.\ref{potupdate} and~\ref{otherAlg}.\ref{newpot}).
Then, combining both effects (the relative increase in potential of $x$ and its neighbors') 
it is $s_{i+1,j,1}^{(2)}[x]\leq s_{i,j,1}^{(2)}[x]=s_{i,j,1}^{(1)}[x]$.
The same argument applies to all other nodes in $L$ with $d/2$ or more neighbors in round $i,j,1$. Thus, it is
$\sum_{x\in L}s^{(2)}_{i+1,j,1}[x] \leq 
\sum_{x\in L}s^{(1)}_{i+1,j,1}[x]$.

Additionally, for any round $\langle i'j',1\rangle$ after round $\langle i,j,1\rangle$, we have to also consider the case of a node $y\in L$ that, although it does not receive at least $d/2$ potentials, it moves to alarm status ``low'' because it has received such status in round $\langle i'j',1\rangle$. 
Then, notice that the potential of $y$ is $\Phi_{i'+1,j',1}[y] = \ell > \tau \geq \Phi_{i',j',1}[y]$ because $y\in L$, and it will stay in $\ell$ for the rest of the epoch (refer to Lines~\ref{leaderAlg}.\ref{alarminsecondleader} and~\ref{otherAlg}.\ref{alarminsecondother}). Therefore, the slack of $y$ is $s_{i'+1,j',1}^{(2)}[y]\leq s_{i'+1,j',1}^{(1)}[y]$, and
$\sum_{y\in L}s^{(2)}_{i'+1,j',1}[y] \leq 
\sum_{y\in L}s^{(1)}_{i'+1,j',1}[y]$.

Combining all the effects studied over all rounds, 
we get the sought
$\sum_{x\in L}s_{0,0,2}[x] \leq \sum_{x\in L}s_{1,1,1}[x] \leq \ell^2$.

\end{proof}


\begin{lemma}
\label{nolowalarm}
Consider an \ADCS 
with a $\cT$-connected evolving graph topology
with $\ell>0$ \ldrs and $n-\ell>0$ \notldr nodes 
running the \algname protocol 
with parameters
$d>k$,
$p>0$,
$r\geq \cT b$, 
$b\geq \left(5+2\epsilon-2\log_k(k^\epsilon-1) \right) \ln k / \phi_{\min}^2$,
and 
$c\geq 5+2\epsilon-2\log_k(k^\epsilon-1)+2\cT+4+\alpha$,
where
$\epsilon>0$,
and
$\alpha \geq \log_k(3\cT)$.
Then, if
$k \geq n$, at the end of the first phase no individual node has potential larger than $\tau=\ell(1-\ell/k^{1+\epsilon})$. \end{lemma}

\begin{proof}
Given that $d\geq n$ and $c > \log_d \cT + 1$, we have from Eq.~\ref{disteq} that right before the beginning of phase $2$ it is
\begin{align*}
\left|\left|\frac{\vec{\Phi}_{0,0,2}}{||\vec{\Phi}_{1,1,1}||_1} - \frac{\vec{I}}{n}\right|\right|_2^2 
&\leq \exp\left(-b\phi_{\min}^2\right)
+ \frac{3\cT}{4d^{c-2\cT-4}}
\ .
\end{align*}

Given that $b\geq \left(5+2\epsilon-2\log_k(k^\epsilon-1) \right) \ln k / \phi_{\min}^2$, it is 
\begin{align*}
\left|\left|\frac{\vec{\Phi}_{0,0,2}}{||\vec{\Phi}_{1,1,1}||_1} - \frac{\vec{I}}{n}\right|\right|_2^2 
&\leq 
\frac{1}{k^{5+2\epsilon-2\log_k(k^\epsilon-1)}}
+ \frac{3\cT}{4d^{c-2\cT-4}}
\ .
\end{align*}

Replacing $\alpha \geq \log_k(3\cT)$ and $d\geq k$ we have that
\begin{align*}
\left|\left|\frac{\vec{\Phi}_{0,0,2}}{||\vec{\Phi}_{1,1,1}||_1} - \frac{\vec{I}}{n}\right|\right|_2^2 
&\leq 
\frac{1}{k^{5+2\epsilon-2\log_k(k^\epsilon-1)}}
+ \frac{1}{4k^{c-2\cT-4-\alpha}}
\ .
\end{align*}

For $c\geq 5+2\epsilon-2\log_k(k^\epsilon-1)+2\cT+4+\alpha$ we have that
\begin{align*}
\left|\left|\frac{\vec{\Phi}_{0,0,2}}{||\vec{\Phi}_{1,1,1}||_1} - \frac{\vec{I}}{n}\right|\right|_2^2 
&\leq 
\frac{5}{4k^{5+2\epsilon-2\log_k(k^\epsilon-1)}}
\leq 
\frac{1}{k^{4+2\epsilon-2\log_k(k^\epsilon-1)}}
\ .
\end{align*}
The latter is true because $k\geq 2$.

Given that for any node $j$, it is $(\Phi_{0,0,2}[j]/||\vec{\Phi}_{1,1,1}||_1 - 1/n)^2 \leq\left|\left|\frac{\vec{\Phi}_{0,0,2}}{||\vec{\Phi}_{1,1,1}||_1} - \frac{\vec{I}}{n}\right|\right|_2^2$, we have that $(\Phi_{0,0,2}[j]/||\vec{\Phi}_{1,1,1}||_1-1/n)^2 \leq 1/k^{4+2\epsilon- 2\ln(k^\epsilon-1)/\ln k}$. Hence, it is  $\Phi_{0,0,2}[j] \leq (1/n + 1/k^{2+\epsilon- \log_k(k^\epsilon-1)})||\vec{\Phi}_{1,1,1}||_1$ for any node $j$.
Moreover, given that $d> k\geq n$ the total potential in the network is $\ell(n-\ell)$ (Claim~\ref{manyconservation}). Thus, no individual node should have potential larger than $\ell(n-\ell)(1/n + 1/k^{2+\epsilon- \log_k(k^\epsilon-1)})$. We show that the latter is at most $\tau=\ell(1-\ell/k^{1+\epsilon})$ as follows. We want to prove
\begin{align*}
\ell(n-\ell)\left(\frac{1}{n} + \frac{1}{k^{2+\epsilon- \log_k(k^\epsilon-1)}}\right) &\leq \ell\left(1-\frac{\ell}{k^{1+\epsilon}}\right)\\
\frac{n-\ell}{k^{2+\epsilon- \log_k(k^\epsilon-1)}} &\leq \frac{\ell}{n}-\frac{\ell}{k^{1+\epsilon}}.
\end{align*}

Given that $k\geq n$, it is enough to show that
\begin{align*}
\frac{k-\ell}{k^{2+\epsilon- \log_k(k^\epsilon-1)}} &\leq \frac{\ell}{k}-\frac{\ell}{k^{1+\epsilon}}\\
k-\ell &\leq \ell (k^\epsilon-1)k^{1- \log_k(k^\epsilon-1)}\\
k-\ell &\leq \ell k \ .
\end{align*}

And the latter is true because $\ell \geq1$.
\end{proof}


The previous lemma shows that, if the estimate is ``not-low'' ($k\geq n$), at the end of the first phase all nodes must have ``low'' potential ($\Phi_{0,0,2}\leq\tau$). 
(Notice the inverse relation between estimate and potential.)
So, to complete the proof of the case $k^{1+\epsilon}<n$ (i.e. low estimate) we show in
the following lemma that if $k^{1+\epsilon}<n$ (i.e. low estimate) there are some nodes with $\Phi_{0,0,2}>\tau$ (i.e. high potential), and that all the other nodes will know this within the following phase.

\begin{lemma}
\label{manyalarmsoon}
Consider an \ADCS 
with a $\cT$-connected evolving graph topology
with $\ell>0$ \ldrs and $n-\ell>0$ \notldr nodes 
running the \algname protocol 
with parameters
$d>k$,
$p>1$
and $r\geq \cT b$,
where
$b \geq \left(5+2\epsilon-2\log_k(k^\epsilon-1) \right) \ln k /
\min\left\{  {\phi_{\min}}^2,\ln(1+i_{\min})\right\}$, and
$\epsilon>0$.
Then, if $k^{1+\epsilon}<n$, 
within the second phase
all nodes receive an alarm status ``low''.
\end{lemma}

\begin{proof}
Consider a partition $\{L,H\}$ of the set of nodes, 
where $L$ is the set of nodes with potential at most $\tau$ at the end of the first phase. 
As shown in Lemma~\ref{manyunalarmed} the size of $L$ is at most $k^{1+\epsilon}$, and because $k^{1+\epsilon}<n$ the size of $H$ is at least $1$.

Based on their ``high'' potential (above $\tau$),
and the property proved in Lemma~\ref{nolowalarm} that in case of not-low estimate, $k\geq n$ there would not be any node like them (notice that $b$ fulfills the condition of such lemma),
all nodes in $H$ 
move to alarm status ``low'' at the end of phase $1$ (refer to Lines~\ref{leaderAlg}.\ref{leaderthreshold} or~\ref{otherAlg}.\ref{otherthreshold}).
(Notice the inverse relation between potential and status, which in turn indicates whether the estimate is low or not.)
We want to compute the number of blocks until every node $L$ has received the low alarm. 

Using the same argument used in Lemma~\ref{lemma:broadcast}, it can be proved that $\ln d / \ln(1+i_{\min})$ blocks are enough to disseminate the alarm throughout the network.
Therefore, within the following $(1+\epsilon)\ln k/\ln(1+i_{\min})$ blocks after the beginning of the second phase any $x\in L$ receives the alarm. Given that $b$ is larger, the claim follows.
\end{proof}


Finally, to complete the proof of correctness, we show in the following lemma that if $k>n$, \ldr nodes detect that the potential consumed is too low for the estimate $k$ to be correct. 

\begin{lemma}
\label{kaboven}
Consider an \ADCS 
with a $\cT$-connected evolving graph topology
with $\ell>0$ \ldrs and $n-\ell>0$ \notldr nodes 
running the \algname protocol 
with parameters $d>k$,
$$p\leq 2\delta\ln k \frac{1-\ell\left(\frac{1}{n}-\frac{1}{k^{\beta}}\right)}{\ell\left(\frac{1}{n}-\frac{1}{k^{\beta}}\right)} \ ,$$
$r\geq \cT b, \ b\geq 4\beta \ln k / \phi_{\min}^2$,
and $c\geq 5\beta+2\cT+4$,
where
$\beta \geq \log_k (n(2k^{\delta} - 1))$, 
$\gamma>\log_k (n-\ell+1)$, and
$$\delta > \log_k \frac{k^{\gamma}(n-\ell)}{k^\gamma-(n-\ell)-1} \ .$$ 
Then, if
$k>n$, 
the potential $\rho$ consumed by any \ldr node is 
$\rho < (k-\ell)\left(1-1/k^\gamma\right)$.
\end{lemma}

\begin{proof}
Given that $d> k > n$ and $c \geq 5\beta+4\cT+2 > \log_d 2\cT +1$, we have from Eq.~\ref{disteq} that right before the beginning of any phase $h>1$ it is
\begin{align*}
\left|\left|\frac{\vec{\Phi}_{0,0,h}}{||\vec{\Phi}_{1,1,h-1}||_1} - \frac{\vec{I}}{n}\right|\right|_2^2 
&\leq \exp\left(-b\phi_{\min}^2\right)
+ \frac{3\cT}{4d^{c-2\cT-4}}.
\end{align*}

For a number of blocks $b\geq 4\beta \ln k / \phi_{\min}^2$, it is
\begin{align*}
\left|\left|\frac{\vec{\Phi}_{0,0,h}}{||\vec{\Phi}_{1,1,h-1}||_1} - \frac{\vec{I}}{n}\right|\right|_2^2 
&\leq \exp\left(-4\beta\ln k\right)
+ \frac{3\cT}{4d^{c-2\cT-4}}
\ = \
\frac{1}{k^{4\beta}}
+ \frac{3\cT}{4d^{c-2\cT-4}}
\ .
\end{align*}

Given that $c\geq 5\beta+2\cT+4$ and $d>k$ we have that
\begin{align*}
\left|\left|\frac{\vec{\Phi}_{0,0,h}}{||\vec{\Phi}_{1,1,h-1}||_1} - \frac{\vec{I}}{n}\right|\right|_2^2 
&\leq \frac{1}{k^{4\beta}}
+ \frac{3\cT}{4d^{5\beta}}
\ .
\end{align*}

And for $d\geq k$ and $\beta\geq \log_k(3\cT)$, it is
\begin{align*}
\left|\left|\frac{\vec{\Phi}_{0,0,h}}{||\vec{\Phi}_{1,1,h-1}||_1} - \frac{\vec{I}}{n}\right|\right|_2^2 
&\leq \frac{5}{4k^{4\beta}}
\leq \frac{1}{k^{2\beta}}\ , \textrm{ for $\beta\geq1/2$ and $k \geq 2$}
\ .
\end{align*}

For any node $j$, given that $(\Phi_{0,0,h}[j]/||\vec{\Phi}_{1,1,h-1}||_1 - 1/n)^2 \leq \left|\left|\frac{\vec{\Phi}_{0,0,h}}{||\vec{\Phi}_{1,1,h-1}||_1} - \frac{\vec{I}}{n}\right|\right|_2^2$ we have that $(\Phi_{0,0,h}[j]/||\vec{\Phi}_{1,1,h-1}||_1 - 1/n)^2 \leq 1/k^{2\beta}$ and hence $\Phi_{0,0,h}[j] \geq (1/n - 1/k^{\beta})||\vec{\Phi}_{1,1,h-1}||_1$.
The latter is true for any initial distribution.
Therefore, after each phase a \ldr node consumes 
at most $1/n+1/k^{\beta}$ fraction of the total
potential in the system, and the total potential in the system drops by at least $\ell(1/n-1/k^{\beta})$ fraction. 
Recall that the initial overall potential in the system is $\ell(n-\ell)$, and that by Claim~\ref{manyconservation}, if $d>n$, the overall potential in the system does not change during each phase.

Using the latter observations, after $p$ phases, any given \ldr node consumes {\em at most}
\begin{align*}
\rho &\leq \ell(n-\ell) \left(\frac{1}{n} + \frac{1}{k^{\beta}}\right) \sum_{i=0}^{p-1} \left(1- \ell\left(\frac{1}{n}-\frac{1}{k^{\beta}}\right)\right)^i.
\end{align*}
Given that $0<\ell\left(\frac{1}{n}-\frac{1}{k^{\beta}}\right)<1$ for $\beta\geq 1$ and $k>n>\ell$, we have that
\begin{align*}
\rho &\leq \ell(n-\ell) \left(\frac{1}{n} + \frac{1}{k^{\beta}}\right) \frac{1-\left(1- \ell\left(\frac{1}{n}-\frac{1}{k^{\beta}}\right)\right)^p}{1-\left(1- \ell\left(\frac{1}{n}-\frac{1}{k^{\beta}}\right)\right)} 
(n-\ell)  \frac{ k^{\beta}+n}{k^{\beta}-n}
\left(1-\left(1- \ell\left(\frac{1}{n}-\frac{1}{k^{\beta}}\right)\right)^p\right)
\ .
\end{align*}

Again using that $0<\ell\left(\frac{1}{n}-\frac{1}{k^{\beta}}\right)<1$ for $\beta\geq 1$ and $k>n>\ell$, we have that
\begin{align*}
\rho &\leq (n-\ell) \frac{ k^{\beta}+n}{k^{\beta}-n}
\left(1-\exp\left(- p\frac{\ell\left(\frac{1}{n}-\frac{1}{k^{\beta}}\right)}{1-\ell\left(\frac{1}{n}-\frac{1}{k^{\beta}}\right)}\right)\right) \ ,
\textrm{ replacing } p\leq 2\delta\ln k \frac{1-\ell\left(\frac{1}{n}-\frac{1}{k^{\beta}}\right)}{\ell\left(\frac{1}{n}-\frac{1}{k^{\beta}}\right)}\ ,\\ 
&\leq (n-\ell) \frac{ k^{\beta}+n}{k^{\beta}-n} \left(1- \frac{1}{k^{2\delta}}\right)
\ = \ 
(n-\ell) \frac{ k^{\beta}+n}{k^{\beta}-n} \left(1+ \frac{1}{k^{\delta}}\right) \left(1- \frac{1}{k^{\delta}}\right)
\ \le \ 
(n-\ell) \left(1+ \frac{1}{k^{\delta}}\right) \ .
\end{align*}
The latter inequality holds for $\beta \geq \log_k (n(2k^{\delta} - 1))$ and $\delta \geq \log_k (3/2)$, 
The second inequality is true because $\log_k \frac{k^{\gamma}(n-\ell)}{k^\gamma-(n-\ell)-1} > \log_k (3/2)$ for $k>n>\ell>0$. 
Then, to complete the proof, it is enough to show that
\begin{align*}
(n-\ell) \left(1+ \frac{1}{k^{\delta}}\right) &< (k-\ell) \left(1- \frac{1}{k^{\gamma}}\right)
\ ,
\end{align*}
which is true for $k>n$, $\delta > \log_k \frac{k^{\gamma}(n-\ell)}{k^\gamma-(n-\ell)-1}$ and $\gamma>\log_k (n-\ell+1)$. 
Hence, the claim follows.
\end{proof}
%
%
%

We establish the correctness and running time of \algname in the following theorem. 

\begin{theorem}
\label{thm:many}
Consider an \ADCS 
with a $\cT$-connected evolving graph topology
with $\ell>0$ \ldrs and $n-\ell>0$ \notldr nodes 
running the \algname protocol 
with parameters: 
\begin{align*}
d &= 2k^{1+\epsilon} \ ,\\ 
p &= \left\lceil \frac{2\ln k}{\ell} \max\left\{ 
\frac{\gamma}{1/k + 1/k^{\alpha}},
\frac{\delta}{1/d+1/k^\beta}
\right\}\right\rceil \ ,\\ 
r &= \lceil\cT b\rceil \ ,\\
\tau &= \ell(1-\ell/k^{1+\epsilon}) \ ,\\
c &\geq 2\cT + 4 + \max\{5\beta , 5 + \alpha + 2\epsilon - 2\log_k(k^\epsilon-1)\} \ ,
\end{align*}
where the number of blocks $b$ is the following. 

\emph{(i)} 
If the \iso $i_{\min}$ is known:
\begin{align}
b =  
\max\left\{\alpha,\beta, 5+2\epsilon-2\log_k (k^\epsilon-1)\right\}
2^{2\cT(2+\epsilon)}
\frac{n^{2\cT(1+\epsilon)}}{i_{\min}^2} 
\ln k.
\label{isopbound}
\end{align}

\emph{(ii)} Otherwise:
\begin{align*}
b =  
\max\left\{\alpha,\beta, 5+2\epsilon-2\log_k (k^\epsilon-1)\right\}
2^{2\cT(2+\epsilon)-2}
n^{2+2\cT(1+\epsilon)}
\ln k.
\end{align*}

Then, under the following conditions:
\begin{align*}
\alpha &\geq \max\{1+\gamma+\log_k 3,\log_k (3\cT)\} \ ,\\ 
\beta &\geq \log_k \max\{d(2k^{\delta} + 1),3\cT\} \ ,\\
\gamma &> \log_k (d-1) \ ,\\ 
\delta &> \log_k \frac{dk^\gamma}{k^{\gamma}+1-d} \ ,\\
\epsilon &> 0 \ ,\\
\cT &\in O(1),
\end{align*}
all nodes stop after at most $\sum_{k\in E\cup B} (pr+d)$ rounds of communication and output $n$,
for $E=\{2^i(\ell+1):i=0,1,\dots,\log\lceil n/(\ell+1)\rceil\}$,
and $B=\{(2^{\log\lceil n/(\ell+1)\rceil}-2^i)(\ell+1):i=0,1,\dots,\log\lceil n/(\ell+1)\rceil-2\}$.
\end{theorem}

\begin{proof}
The proof of correctness of the computation and running time for Part \emph{(i)} is similar to~\cite{KowalskiMicalp19journal}, adapted to our parameters and applying the bound on conductance in Eq.~\ref{condvsisop}, 
and Part \emph{(ii)} is obtained applying the lower bound $i_{\min}\geq 2/n$ to Equation~\ref{isopbound}.

About the \ADCS limitations, messages sent by nodes are only their status and potential. The status requires only 4 bits. Potentials are truncated to $c \log d \leq c(1+(1+\epsilon)\log (2n))$ bits (refer to Lines~\ref{leaderAlg}.\ref{potupdate} and~\ref{otherAlg}.\ref{newpot})
fulfilling the restrictions on message size and memory access at the same time as long as there exists a $c\in O(1)$, which can be seen replacing tight bounds on the conditions above. 
\end{proof}

In the following corollary we relate the knowledge of network characteristics to the asymptotic running time of~\algname. 

\begin{corollary}
\label{cor:time}
The time complexity of \algname on an \ADCS with $\ell>0$ \ldr nodes and $n-\ell>0$ \notldr nodes is the following.

\emph{(i)} 
If the \iso $i_{\min}$ is known:
\begin{align*}
\widetilde{O}\left( 
 \frac{n^{1+2\cT(1+\epsilon)}}{\ell i_{\min}^2}
\right).
\end{align*}

\emph{(ii)} Otherwise:
\begin{align*}
\widetilde{O}\left( 
\frac{n^{3+2\cT(1+\epsilon)}}
{\ell} 
\right).
\end{align*}

\end{corollary}
\begin{proof}
Fixing $\gamma = \log_k d$ it would be $\log_k \frac{dk^\gamma}{k^{\gamma}+1-d} = \log_k d^2 = 2(\log_k 2 + 1+\epsilon)$. That is, it is enough to set $\gamma = \log_k d$ and $\delta=2(2+\epsilon)$ to meet the conditions on those parameters in Theorem~\ref{thm:many}. Replacing these and $d=2k^{1+\epsilon}<2(2n)^{1+\epsilon}$,
\begin{align*}
p &< \left\lceil \frac{2d\ln d}{\ell} \right\rceil
\in O\left(\frac{n\ln n}{\ell}\right)
\ .
\end{align*}
Replacing $\gamma$ and $\delta$ in $\alpha$ and $\beta$, and $\cT\in O(1)$ it is
\begin{align*}
r 
&\in O\left( 
\frac{1}
{\min\left\{  {\phi_{\min}}^2,\ln(1+i_{\min})\right\}}\right)
\ .
\end{align*}
Then, it is
\begin{align*}
pr + d &\in O\left( 
\frac{n\ln n }
{\ell \min\left\{  {\phi_{\min}}^2,\ln(1+i_{\min})\right\}}\right)
\ .
\end{align*}
The total number of terms in the summation of the running time in
Theorem~\ref{thm:many} is $O(\log \frac{n}{\ell})$, 
hence
the claim follows.
\end{proof}



\section{Discussion and Open Problems}

Although the presented algorithm is the first that guarantees a polynomial emulation of \allallfull on the top of anonymous congested highly-dynamic (and not necessarily always connected) systems, the main challenges are to further shrink the polynomials and/or provide lower bounds better than long-time known $\Omega(D+n\log n)$, where $D$ stands for temporal diameter.
We hypothesize that shrinking this complexity gap may depend on computational power, for instance, being able to process history trees may allow provably faster \allallfull.


\bibliographystyle{plain}
\bibliography{Comprehensive_2010}

\end{document}